\documentclass[conference]{IEEEtran}
\IEEEoverridecommandlockouts

\IEEEaftertitletext{\vspace{1.1\baselineskip}}
\usepackage{cite}
\usepackage{amsmath,amssymb,amsfonts}
\usepackage{graphicx}
\usepackage{textcomp}
\usepackage{xcolor}
\def\BibTeX{{\rm B\kern-.05em{\sc i\kern-.025em b}\kern-.08em
    T\kern-.1667em\lower.7ex\hbox{E}\kern-.125emX}}
\usepackage{graphicx} 
\usepackage{amsmath,amssymb,amsfonts}
\usepackage{textcomp}
\usepackage{xcolor}
\usepackage {lipsum} 
\usepackage {amsthm}
\usepackage{colortbl}
\usepackage{caption}
\usepackage {algorithm}
\usepackage{algpseudocode}
\usepackage{tikz}
\usetikzlibrary{arrows.meta}
\definecolor{Gray}{gray}{0.85}
\newcolumntype{a}{>{\columncolor{Gray}}c}
\theoremstyle{definition} 
\newtheorem{Example}{Example}
\newtheorem{Remark}{Remark}
\newtheorem{Lemma}{Lemma}

\newtheorem{Theorem}{Theorem}
\newtheorem{Definition}{Definition}

\newtheorem{Corollary}{Corollary}

\begin{document}
\title{Secure Multi-User Linearly-Separable \\Distributed Computing
\thanks{The work of Amir Masoud Jafarpisheh is supported by UK Research and Innovation (UKRI) under the UK government’s Horizon Europe funding Guarantee under grant EP/Z536404/1, as part of the FOCAL project funded under Marie Skłodowska-Curie grant agreement No 101169042. This work was also supported by the Huawei France-funded Chair towards Future Wireless Networks, and by the French government under the France 2030 ANR program “PEPR Networks of the Future” (ref. ANR-22-PEFT-0010).\\
A. Jafarpisheh is with the Institute for Imaging, Data, and Communications (IDCOM), School of Engineering, University of Edinburgh, Edinburgh, UK.
A. Khalesi is with the Institut Polytechnique
des Sciences Avancées (IPSA) and LINCS Lab, Paris, France.  P. Elia is with the Communication Systems Department, EURECOM, Sophia Antipolis, France.}}

\author{\IEEEauthorblockN{Amir Masoud Jafarpisheh}{
am.jafarpisheh@ed.ac.uk}
\and
\IEEEauthorblockN{Ali Khalesi}{
ali.khalesi@ipsa.fr}
\and
\IEEEauthorblockN{Petros Elia}{
elia@eurecom.fr}
}

\maketitle
\begin{abstract}
The introduction of the new multi-user linearly-separable distributed computing framework, has recently revealed how a parallel treatment of users can yield large parallelization gains with relatively low computation and communication costs. These gains stem from a new approach that converts the computing problem into a sparse matrix factorization problem; a matrix \(\mathbf{F}\) that describes the users' requests, is decomposed as \(\mathbf{F} = \mathbf{DE}\), where a \(\gamma\)-sparse \(\mathbf{E}\) defines the task allocation across \(N\) servers, and a \(\delta\)-sparse \(\mathbf{D}\) defines the connectivity between \(N\) servers and \(K\) users as well as the decoding process. While this approach provides near-optimal performance, its linear nature has raised data secrecy concerns. 

We adopt an information-theoretic secrecy framework requiring that each user learns nothing more than its own requested function. Our main results provide (i) a necessary condition stating that for each user $k$ observing $\alpha_k$ server responses, the common randomness visible to that user must span a subspace of dimension greater than $\alpha_k-1$, and (ii) a necessary and sufficient condition requiring that removing from $\mathbf{D}$ the columns corresponding to the servers observed by a user leaves a matrix of rank at least $K-1$. Based on these conditions, we design a general, cost-preserving secrecy-enforcing transformation valid over both finite and real fields, obtained by appending to $\mathbf{E}$ a basis of $\mathrm{Null}(\mathbf{D})$ and carefully injecting shared randomness. This scheme preserves communication and computation costs, guarantees perfect information-theoretic secrecy over finite fields, and in the real case yields an explicit mutual-information bound that can be made arbitrarily small by increasing the variance of Gaussian common randomness.
\end{abstract}

\section{Introduction}
Distributed data analytics increasingly operate across administrative and trust boundaries, where multiple users are authorized to compute functions of shared datasets while the underlying data must remain confidential. Prominent examples arise in multi-tenant cloud services, federated learning, healthcare analytics, financial risk assessment, and smart-grid monitoring. In these systems, users request distinct linear functionals of the same data, yet any information beyond the authorized computation may lead to severe privacy and data secrecy breaches.

A subtle but fundamental challenge arises even in purely linear workflows. While linear functions are among the simplest classes of computations, the ability to recover \emph{additional} unintended linear combinations can enable strong inference attacks through aggregation, side information, or repeated queries. This issue persists in both discrete computations over $\mathbf{GF}(q)$ and real-valued computations over $\mathbb{R}$, where privacy guarantees based on computational hardness may fail due to auxiliary information or numerical precision. These considerations motivate an information-theoretic treatment of secrecy that provides unconditional guarantees independent of computational assumptions. 

\subsection{Linearly Separable Computation}
We restrict attention to \emph{linearly separable} distributed computing schemes, in which each user’s requested function can be expressed as a linear combination of intermediate computations performed independently at the servers. Each user $k \in [K]$ requests a linear combination of the messages, $\langle \mathbf{f}_k, \mathbf{w} \rangle,$
where $\mathbf{f}_k \in \mathbb{F}^L$ is the request vector of user $k$, and $\mathbf{w} = [W_1, \ldots, W_L]^{\mathsf{T}}$ denotes the message vector (subfunction outputs) computed across the various servers. 

For $\mathbf{F} = [\mathbf{f}_1, \mathbf{f}_2, \ldots, \mathbf{f}_K]^{\mathsf{T}} \in \mathbb{F}^{K \times L}$ denoting the users’ request matrix, we know from \cite{khalesi2023multi} that our distributed computing problem can be mapped into a matrix factorization problem, where \(\mathbf{F}\) is factorized as
$\mathbf{F} = \mathbf{D}\mathbf{E},$
where $\mathbf{E} = [\mathbf{e}_1, \mathbf{e}_2, \ldots, \mathbf{e}_N]^{\mathsf{T}} \in \mathbb{F}^{N \times L},$ specifies the linear computations assigned to the servers and where $\mathbf{D} = [\mathbf{d}_1, \mathbf{d}_2, \ldots, \mathbf{d}_K]^{\mathsf{T}} \in \mathbb{F}^{K \times N},$ specifies the server-to-user connectivity and the linear aggregation performed by the users. 

In this same context, each server \(n\) computes 
\(A_n = \mathbf{e}_n^{\mathsf{T}}\mathbf{w},\) 
while user \(k\) recovers its desired value as \(\mathbf{d}_k^{\mathsf{T}}\mathbf{A} = \mathbf{f}_k^{\mathsf{T}}\mathbf{w}\), where $\mathbf{A} = [A_1, A_2, \ldots, A_N]^{\mathsf{T}}$. This algebraic structure encompasses a broad class of distributed inference and analytics tasks, and enables a precise characterization of correctness, efficiency~\cite{khalesi2023multi}, and as we will see here, information-theoretic secrecy as well.

\subsection{Related Works}
Distributed computing systems are fundamentally constrained by limited server computation and finite communication capacity, giving rise to the fundamental communication–computation tradeoff that pervades distributed systems. This same tradeoff has been the focus of information-theoretic expositions in various settings such as~\cite{li2017fundamental,yan2022storage,10947203,malak2024multi} to mention just a few. 
In addition, coding has also emerged as a powerful tool, not only for reducing communication costs, but also for mitigating stragglers in linear computations settings~\cite{tandon2017gradient,ye2018communication,yu2020straggler,11195435,10458969}. 
For example, the early work in~\cite{wan2021distributed} considered the problem of distributed computation of linearly separable functions, and focusing on the single-user regime, proceeded to derive optimality results under the assumption of cyclic task assignments with an emphasis on straggler resilience. Other interesting directions that followed, with novel schemes and bounds, can be found in~\cite{wan2021tradeoff,namboodiri2025fundamental} for the single-user regime, as well as in~\cite{khalesi2023multi,namboodiri2026fundamental, 10619542, 10947203, khalesi2023multi2} for the multi-user variant of this problem.  This multiuser approach, which is indeed the closest to our current setting, considered the setting where each server serves multiple users, each with independent requests. These multiuser variants employ novel techniques based on covering codes and tessellation-based task assignments, respectively, but focus exclusively on communication and computation optimality, without addressing information leakage. Thus, while indeed, these prior results establish near-optimal communication and computation policies, they do not determine whether users inevitably obtain unintended information as a consequence of the decoding structure.

Such data secrecy guarantees are important, and have sparked a long line of related research, within the context of distributed computing. 
For example, private function retrieval and private computation frameworks~\cite{sun2018capacity,mirmohseni2018private,10891903} entail novel techniques that carefully conceal users’ requested functions from servers, while secure coded computing schemes incorporate resiliency and privacy against stragglers and adversaries~\cite{yu2019lagrange}. 
Furthermore, private access control~\cite{jafarpisheh2022distributed, meel2024hetdapac} as well as, secret-sharing–based approaches provide information-theoretic data protection~\cite{shamir1979share,bitar2020minimizing, tjell2021privacy}, and multi-user secret sharing~\cite{khalesi2021capacity} extends this latter discipline to the setting of multiple users/receivers. 
Perhaps closer to our own setting is the approach in~\cite{wan2022secure} which involves the study of information-theoretic data secrecy in distributed linearly-separable computation, primarily for single-user systems and specific data assignment structures.

Despite these advances, existing approaches remain tied to particular encoding schemes, task assignments, or cryptographic primitives. Consequently, a general information-theoretic characterization of data secrecy that depends on the structure of the decoding process has not yet been fully established. Focusing on the promising multi-user linearly separable setting, and here seek to understand under what structural conditions can each user provably recover only its requested function output, with no additional information being leaked. Our approach will reveal conditions that span across fields and which are not restricted to specific encoding schemes.

\subsection{Main Contributions}
To address the above question, our work here provides the following contributions. 
\begin{itemize}
    \item \textbf{Per-user necessary secrecy condition:}
We show that if user $k$ observes $w_{\mathsf{H}}(\mathbf{d}^{\mathsf{T}}_k)$\footnote{For a vector $\mathbf{x}$, the Hamming weight $w_{\mathsf{H}}(\mathbf{x})$ is defined as the number of nonzero components of $\mathbf{x}$.} server responses, and the common randomness \emph{visible to that user} spans a subspace of dimension less than $w_{\mathsf{H}}(\mathbf{d}^{\mathsf{T}}_k)-1,$ data secrecy (as it will be defined clearly later on) cannot be guaranteed.
This condition implies the access constraint $w_{\mathsf{H}}(\mathbf{d}^{\mathsf{T}}_k)\le N-K+1$ and yields the universal converse
$\delta \le 1-\frac{K-1}{N}$ on the communication cost. 
\item \textbf{Decoding-matrix-only secrecy criterion:}
Data secrecy holds if and only if, for each user \(k\), removing from \(\mathbf{D}\) the columns corresponding to the servers observed by that user reduces the rank by at most one. Over \(\mathbf{GF}(q)\), this condition guarantees perfect information-theoretic secrecy (zero leakage), while over \(\mathbb{R}\) it guarantees a bound on the mutual-information leakage that can be made arbitrarily small by injecting Gaussian common randomness. Moreover, the total leakage grows only proportionally to the per-response information rate and the number of observable responses, yielding a transparent and scalable data secrecy guarantee.
\item \textbf{Cost-preserving secure transformation:} 
    Given any admissible (possibly non-secure) factorization $\mathbf{F}=\mathbf{D}\mathbf{E}$, we construct a secure scheme by appending to \(\mathbf{E}\) a basis for the null-space \(\mathrm{Null}(\mathbf{D})\) and embedding shared common randomness along these null-space directions. Over \(\mathbf{GF}(q)\),
perfect information-theoretic
secrecy (zero leakage) is guaranteed. Over $\mathbb{R},$ data secrecy is achieved by turning every unintended decoding direction into a channel with low signal to noise ratio (as it will be defined later on). Correctness is preserved, and both communication and computation costs remain unchanged under the standard cost model.
\end{itemize}

Unlike existing secure coded-computing constructions that are tailored to specific encoding solutions or to secret sharing~\cite{yu2019lagrange, bitar2020minimizing, tandon2017gradient,ye2018communication,yu2020straggler, 11195435,10458969}, 
our work provides necessary and sufficient structural conditions that focuses entirely on the decoding matrix, as well as provides a universal, cost-preserving secrecy-enforcing transformation applicable to any admissible linearly separable multi-user scheme.

In the following, Section~\ref{section_system_model} introduces the system model for secure multi-user linearly separable distributed computing, while Section~\ref{section_main_results} presents the main results, with the achievable design being described in Section~\ref{Achievable_algorithm}, and the conclusions being drawn in Section~\ref{section_conclusions}.

\section{System Model}
\label{section_system_model}
\begin{figure}[tb]
\vspace{13pt}
\centering
\includegraphics[trim={.5cm .6cm .5cm .65cm},clip,scale=.72]{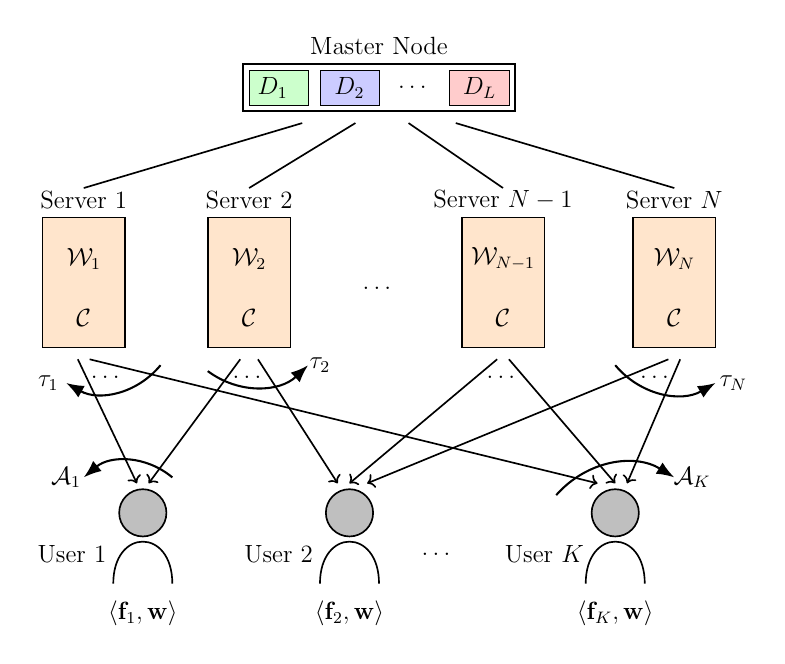}
\captionsetup{justification=centering}
\caption{System model of secure multi-user linearly-separable distributed computing}
\label{system_model_fig}
\end{figure}

As illustrated in Fig.~\ref{system_model_fig}, we consider an $(N, K, L)$ secure multi-user distributed computing system consisting of $N$ servers, $K$ users, and $L$ independent messages $W_1, W_2, \ldots, W_L$, coordinated by a master node. Suppose $N\geq K$, $L\geq K$, and for each $l \in [L]$, each such message $W_l = f_l(D_l)$ is the output of a function $f_l(\cdot)$ that is generally a non-linear and computationally hard function of input data $D_l$. Furthermore, the master node serves as an orchestrator, coordinating the interactions between the users and the servers.

In our setting, each user $k \in [K]$ requests an independent linear combination $\langle \mathbf{f}_k, \mathbf{w} \rangle$ of the message vector $\mathbf{w} = [W_1, \ldots, W_L]^{\mathsf{T}}$, where $\mathbf{f}_k \in \mathbb{F}^L$ is the \emph{request vector} defining the request of user $k$. These users’ request vectors are assumed to be independent of the messages, in the sense that 
\begin{align}
    I(\mathbf{f}_1, \mathbf{f}_2, \ldots, \mathbf{f}_K; \mathbf{w}) = 0.
    \label{query_indepence}
\end{align}

The messages are independent and identically distributed --- either uniformly over $\mathbf{GF}(q)$ or according to $\mathcal{N}(0,\sigma_w^2)$ in the real-valued case --- yielding
\begin{align}
    H(\mathbf{w}) = \sum\limits_{i=1}^L H(W_i)= LH(W_1).
\end{align}

Following the notation in \cite{khalesi2023multi}, upon receiving the users’ request vectors $\mathbf{f}_1, \mathbf{f}_2, \ldots, \mathbf{f}_K$, the master node constructs the full-rank \emph{request matrix}
\begin{align}
    \mathbf{F} = [\mathbf{f}_1, \mathbf{f}_2, \ldots, \mathbf{f}_K]^{\mathsf{T}} \in \mathbb{F}^{K \times L},
\end{align}
and performs a matrix factorization $\mathbf{F} = \mathbf{D}\mathbf{E}$. This factorization yields the \emph{encoding matrix}
\begin{align}
    \mathbf{E} = [\mathbf{e}_1, \mathbf{e}_2, \ldots, \mathbf{e}_N]^{\mathsf{T}} \in \mathbb{F}^{N \times L},
\end{align}
which first specifies the computation tasks assigned to the servers, and how each server combines these computed outputs~- as specified by $\mathbf{e}_n^{\mathsf{T}}$ for each server $n \in [N]$. 
The factorization also specifies the full-rank \emph{decoding matrix}
\begin{align}
    \mathbf{D} = [\mathbf{d}_1, \mathbf{d}_2, \ldots, \mathbf{d}_K]^{\mathsf{T}} \in \mathbb{F}^{K \times N},
\end{align}
which determines how each user aggregates the servers’ responses, with $\mathbf{d}_k^{\mathsf{T}}$ denoting the aggregation vector of user $k$.

As each server $n \in [N]$ is notified, by the master node, of its encoding vector $\mathbf{e}_n^{\mathsf{T}}$ and of its broadcast schedule $\tau_n$, it proceeds to compute a deterministic function of a subset of messages $\mathcal{W}_n \subseteq \mathcal{W}=\{W_1,\ldots, W_L\}$, possibly using shared independent common randomness $\mathcal{C}=\{C_1, C_2,\ldots, C_m\}$ with sufficiently large i.i.d elements. This randomness is shared among the servers, is independent of the messages, and is used later to guarantee data secrecy. Naturally, for each server $n$, we have
\begin{align}
    H(A_n \mid \mathcal{W}_n, \mathcal{C}, \mathbf{e}_n) = 0,
    \label{answer_func}
\end{align}
where we recall that $A_n = \mathbf{e}_n^{\mathsf{T}}\mathbf{w}$ is what is broadcast by server $n$ to the users in $\tau_n$.\footnote{Our framework imposes a one-shot constraint: each server is allowed to send only one linear combination, called a response, to its assigned users.}

Now, each user $k \in[K]$ receives the set of responses 
\begin{align} \label{eq:Ak1}
\mathcal{A}_k=\{A_n: n\in\mathrm{Sup}(\mathbf{d}_k^{\mathsf{T}})\} ,
\end{align}
from a subset of servers indexed by the support $\mathrm{Sup}(\mathbf{d}_k^{\mathsf{T}}) \subseteq [N]$ of their own decoding vector $\mathbf{d}_k^{\mathsf{T}}$. After receiving from the servers, each user $k$ must be able to recover its requested linear combination using $\mathbf{d}_k^{\mathsf{T}}$. 

This requirement is captured by the following \emph{correctness condition}
\begin{align}
    H(\langle \mathbf{f}_k, \mathbf{w} \rangle \mid \mathcal{A}_k, \mathbf{d}_k, \mathbf{f}_k) = 0.
    \label{correctness_condition}
\end{align}
Furthermore, user $k$ must learn nothing about the messages more than its requested linear combination. This is formalized by the \emph{data secrecy constraint}. In particular, in the case of $\mathbf{GF}(q)$, this data secrecy constraint is captured by
\begin{align}
    I(\mathcal{W}; \mathcal{A}_k, \mathbf{f}_k \mid \langle \mathbf{f}_k, \mathbf{w} \rangle) = 0,
    \label{data_secrecy_condition}
\end{align}
while for the case of real-valued computations, data secrecy is satisfied if for any $\varepsilon > 0$, there exists a variance $\sigma_c^2$ for the i.i.d.\ Gaussian common random variables with $C_i \sim \mathcal{N}(0,\sigma_c^2)$, such that
\begin{align}
    I(\mathcal{W}; \mathcal{A}_k, \mathbf{f}_k \mid \langle \mathbf{f}_k, \mathbf{w} \rangle) \leq \varepsilon.
    \label{Privacy_con}
\end{align}
We briefly summarize that an $(N,K,L)$ secure distributed computing scheme entails the users’ request vectors $\{\mathbf{f}_k\}_{k=1}^K$, the servers’ encoding vectors $\{\mathbf{e}_n\}_{n=1}^N$, the users’ decoding vectors $\{\mathbf{d}_k\}_{k=1}^K$, and the shared common randomness $\mathcal{C}$.

Finally, as is common, the \textit{communication cost} is defined as
\begin{align}
    \delta = \frac{\sum\limits_{n=1}^N |\tau_n|}{K N}
           = \frac{\sum\limits_{k=1}^K w_{\mathsf{H}}(\mathbf{d}^{\mathsf{T}}_k)}{K N},
    \label{delta_def}
\end{align}
while the \emph{computation cost} is defined as
\begin{align}
    \gamma = \frac{\max\limits_{l \in [L]} w_{\mathsf{H}}(\mathbf{E}(:,l))}{N},
    \label{gamma_def}
\end{align}
where, for each $l \in [L]$, $w_{\mathsf{H}}(\mathbf{E}(:,l))$ denotes the number of servers that compute message $W_l$. With these in place, we have the following definition. 
\begin{Definition}
A pair $(\gamma,\delta)$ is said to be \emph{feasible} if there exists a secure distributed computing scheme with computation cost $\gamma$ and communication cost $\delta$ that satisfies the correctness condition in \eqref{correctness_condition} as well as the data secrecy constraint in either \eqref{data_secrecy_condition} or \eqref{Privacy_con}, depending on the field case.\footnote{In this paper, we exclude the trivial case in which each user can directly access specific components of the common randomness $\mathcal{C}$ via a link to the master node, which would trivially guarantee data secrecy in the multi-user linearly-separable distributed computing problem.}
\end{Definition}

To further illustrate the multi-user linearly separable distributed computing setting, we present the following simple example.

\begin{Example}
\label{Discrete_Ex}
Consider an \((N = 6,K= 4,L =5)\) multi-user linearly separable distributed computing system with \(N=6\) servers, \(K=4\) users, and \(L=5\) independent messages \(W_1, W_2, \ldots, W_5\), where we recall that for each \(l \in [5]\), each
$W_l = f_l(D_l),$
is simply the output of a function \(f_l(\cdot)\) having input the dataset \(D_l\).

Let us also consider that each user $k \in [4]$ requests an independent linear combination $\langle \mathbf{f}_k, \mathbf{w} \rangle$ of the messages $\mathbf{w} = [W_1, W_2, \ldots, W_5]^{\mathsf{T}}$, where here $\mathbf{f}_1 = [3, 0, -3, 4, -1]^{\mathsf{T}}$,
$\mathbf{f}_2 = [0, 0, 2, 6, 1]^{\mathsf{T}},
\mathbf{f}_3 = [0, 3, 1, 3, 1]^{\mathsf{T}},$ and $
\mathbf{f}_4 = [3, 3, 6, 1, 1]^{\mathsf{T}}.$

After forming the request matrix 
\[
\mathbf{F} =
\begin{bmatrix}
3 & 0 & -3 & 4 & -1\\
0 & 0 & 2 & 6 & 1 \\
0 & 3 & 1 & 3 & 1 \\
3 & 3 & 6 & 1 & 1 
\end{bmatrix},
\]
the master node performs matrix factorization
$\mathbf{F} = \mathbf{D}\mathbf{E}$ in the following manner
\[
\mathbf{D} =
\begin{bmatrix}
2 & 1 & -1 & 0 & 0 & 0\\
0 & 1 & 1 & 1 & 0 & 0\\
0 & 0 & 1 & 0 & 3 & 0\\
0 & 0 & 0 & 1 & 1 & 1
\end{bmatrix},
\:\:
\mathbf{E} =
\begin{bmatrix}
1 & 0 & 0 & 2 & 0 \\
1 & 0 & -2 & 3 & 0 \\
0 & 0 & 1 & 3 & 1 \\
-1 & 0 & 3 & 0 & 0 \\
0 & 1 & 0 & 0 & 0 \\
4 & 2 & 3 & 1 & 1
\end{bmatrix}.
\]
Based on $\mathbf{E} = [\mathbf{e}_1,\ldots,\mathbf{e}_6]^{\mathsf{T}}$ and
$\mathbf{D} = [\mathbf{d}_1,\ldots,\mathbf{d}_4]^{\mathsf{T}}$, 
the master node assigns computation tasks and broadcast sets
$\{\tau_n\}_{n=1}^6$ to the servers. In particular, each server $n$ computes $A_n = \mathbf{e}_n^{\mathsf{T}}\mathbf{w}$ and
broadcasts it to users in $\tau_n$, as summarized in
Table~\ref{server_comp}, and each user $k \in [4]$
\begin{table}[tb]
\vspace{7pt}
\caption{Server-side computations in Example~\ref{Discrete_Ex}.}
\centering
\renewcommand{\arraystretch}{1.3}
\begin{tabular}{|c|c|c|}
\hline
Server & Computed linear combination $A_n =\mathbf{e}_n^{\mathsf{T}}\mathbf{w}$ & Broadcast set $\tau_n$ \\ \hline
1 & $W_1 + 2W_4$ & $\{1\}$ \\ \hline
2 & $W_1 - 2W_3 + 3W_4$ & $\{1,2\}$ \\ \hline
3 & $W_3 + 3W_4 + W_5$ & $\{1,2,3\}$ \\ \hline
4 & $-W_1 + 3W_3$ & $\{2,4\}$ \\ \hline
5 & $W_2$ & $\{3,4\}$ \\ \hline
6 & $4W_1 + 2W_2 + 3W_3 + W_4 + W_5$ & $\{4\}$ \\ \hline
\end{tabular}
\label{server_comp}
\end{table}
accesses servers in $\mathrm{Sup}(\mathbf{d}_k^{\mathsf{T}})$ and computes
$\mathbf{d}_k^{\mathsf{T}}\mathbf{A}=\mathbf{d}_k^{\mathsf{T}}\mathbf{E}\mathbf{w}$, with $\mathbf{A} = [{A}_1, {A}_2, \ldots, {A}_6]^{\mathsf{T}}$, as shown in
Table~\ref{user_comp}.

\begin{table}[tb]
\caption{User-side aggregation in Example~\ref{Discrete_Ex}.}
\centering
\renewcommand{\arraystretch}{1.3}
\begin{tabular}{|c|c|c|}
\hline
User & Accessed servers $\mathrm{Sup}(\mathbf{d}_k^{\mathsf{T}})$ & Computation $\mathbf{d}_k^{\mathsf{T}}\mathbf{A}$ \\ \hline
1 & $\{1,2,3\}$ & $2A_1 + A_2 - A_3$ \\ \hline
2 & $\{2,3,4\}$ & $A_2 + A_3 + A_4$ \\ \hline
3 & $\{3,5\}$ & $A_3 + 3A_5$ \\ \hline
4 & $\{4,5,6\}$ & $A_4 + A_5 + A_6$ \\ \hline
\end{tabular}

\label{user_comp}
\end{table}
At this point, we can directly see that each user $k$ can correctly recover its desired linear combination since
$\mathbf{d}_k^{\mathsf{T}}\mathbf{A}
= \mathbf{d}_k^{\mathsf{T}}\mathbf{E}\mathbf{w}
= \mathbf{f}_k^{\mathsf{T}}\mathbf{w}.$
Using the communication and computation cost metrics~\eqref{delta_def} and~\eqref{gamma_def}, the resulting costs are
$$\delta = \frac{\sum_{n=1}^{N} |\tau_n|}{K N} = \frac{11}{24}, \gamma = \frac{\max_{l \in [L]} w_{\mathsf{H}}\!\left(\mathbf{E}(:,l)\right)}{N} = \frac{4}{6},$$
where $w_{\mathsf{H}}\!\left(\mathbf{E}(:,l)\right)$ denotes the number of servers computing message $W_l$.
\end{Example}

In Example~\ref{Discrete_Ex}, each user $k \in [4]$ has access to a set of responses $\mathcal{A}_k$. Specifically,
\begin{align}
\mathcal{A}_1 &= \{W_1+2W_4,\; W_1-2W_3+3W_4,\; W_3+3W_4+W_5\},\nonumber\\
\mathcal{A}_2 &= \{W_1-2W_3+3W_4,\; W_3+3W_4+W_5,\; -W_1+3W_3\},\nonumber\\
\mathcal{A}_3 &= \{W_3+3W_4+W_5,\; W_2\},\nonumber\\
\mathcal{A}_4 &= \{-W_1+3W_3,\; W_2,\; 4W_1+2W_2+3W_3+W_4+W_5\}.\nonumber
\end{align}

The collection $\mathcal{A}_k$ generally spans a subspace of dimension greater than one. Consequently, user $k$ may form additional linear combinations beyond its request from $\mathrm{Span}(\mathcal{A}_k)$, which raises the question of whether one can provably guarantee that each user can recover exactly one authorized linear functional of the data, and no other additional information.  This question is addressed in the following.

\section{Main Results}
\label{section_main_results}
Lemma~\ref{rank_lemma} establishes a necessary condition for data secrecy, leading to a converse bound on the communication cost \(\delta\) in Theorem~\ref{Theorem_necessary_condition}. Theorem~\ref{sufficient_condition} provides a necessary and sufficient condition for data secrecy over both \(\mathbb{R}\) and \(\mathbf{GF}(q)\).

As one would expect, the resulting solution will entail the use of common randomness. Related to this, let us consider \textit{the common randomness coefficient matrix} $\mathbf{C}$ that collects the coefficients associated with the common randomness variables used by the servers. For each user $k\in[K]$ connected to the $\mathrm{Sup}(\mathbf{d}_k^{\mathsf{T}})$ servers, let $\mathbf{C}(\mathrm{Sup}(\mathbf{d}_k^{\mathsf{T}}),:)$ denote the submatrix formed by the rows of $\mathbf{C}$ corresponding to the servers in $\mathrm{Sup}(\mathbf{d}_k^{\mathsf{T}})$. Moreover, assume that the submatrix $\mathbf{E(\mathrm{Sup(\mathbf{d}_k^{\mathsf{T}}),:})}$ is full rank.\footnote{This is a mild non-degeneracy condition on $\mathbf{E}.$ Without it, some server responses carry redundant message mixtures.
} Lemma~\ref{rank_lemma} provides a necessary condition on $\mathbf{C}(\mathrm{Sup}(\mathbf{d}_k^{\mathsf{T}}),:)$ to guarantee data secrecy for each user $k$.

\begin{Lemma}
\label{rank_lemma}
In any multi-user linearly separable distributed computing scenario, where each user $k \in [K]$ is connected to $w_{\mathsf{H}}(\mathbf{d}^{\mathsf{T}}_k)$ servers, if 
\begin{align}
    \mathrm{Rank}\!\big(\mathbf{C}(\mathrm{Sup}(\mathbf{d}_k^{\mathsf{T}}),:)\big) < w_{\mathsf{H}}(\mathbf{d}^{\mathsf{T}}_k) - 1,
    \label{Rank_requirement}
\end{align}
\end{Lemma}
then data secrecy cannot be guaranteed for user $k$. 

\begin{IEEEproof}
User $k$ observes $w_{\mathsf{H}}(\mathbf{d}^{\mathsf{T}}_k)$ linear 
combinations of the messages and the common randomness. Data secrecy requires that no additional information about the messages be revealed. This is possible only if at most one degree of freedom remains after canceling the common randomness, which implies
$\mathrm{Rank}\!\big(\mathbf{C}(\mathrm{Sup}(\mathbf{d}_k^{\mathsf{T}}),:)\big)<w_{\mathsf{H}}(\mathbf{d}^{\mathsf{T}}_k)-1.$
A detailed proof is given in Appendix~\ref{proof_rank_user}.
\end{IEEEproof}

We now proceed with the following theorem which provides a converse bound on the communication cost~$\delta$, over which it is impossible to induce data secrecy.

\begin{Theorem}
\label{Theorem_necessary_condition}
Consider any $(N,K,L)$ multi-user linearly separable distributed computing scheme with $N \geq K$.
If the communication cost satisfies
\[
\delta > 1 - \frac{K-1}{N},
\]
then there exists at least one user that can obtain information beyond its requested linear combination.
\end{Theorem}

\begin{IEEEproof}
Using Lemma~\ref{rank_lemma}, we can derive the converse bound $w_{\mathsf{H}}(\mathbf{d}_k^{\mathsf{T}}) > N - K + 1$, which limits the number of servers accessible to user $k$ over which data secrecy is impossible. The result follows by applying Lemma~\ref{rank_lemma} to all \(K\) users and using~\eqref{delta_def}. A detailed proof is given in Appendix~\ref{proof_delta_bound}.
\end{IEEEproof}

The following remark, illustrates how the converse bound on the communication cost $\delta$ also translates into the computation cost~$\gamma$. 
\begin{Remark}
Consider an $(N, K,L)$  multi-user linearly separable distributed computing scheme with $N \geq K$ and $L=K$, in the fully decentralized case
$\mathbf{F} = \mathbf{D}\mathbf{I}_N,$ and $\mathbf{E} = \mathbf{I}_{N}.$ Thus, each server $n \in [N],$ computes only function $W_n=f_n(D_n)$, yielding $\gamma = \frac{1}{N}$. In this scenario, $\mathbf{D}= \mathbf{F}$, and in general, $\delta > 1 - \frac{K-1}{N}.$
So at least one user can infer more than its requested linear combination.
This illustrates that servers must perform a sufficient amount of computation to mix the messages effectively; in general, if the servers' computational capabilities are too limited, data secrecy cannot be guaranteed.
\end{Remark}

The following theorem holds over either $\mathbb{R}$ or $\mathbf{GF}(q)$ and establishes a necessary and sufficient condition on the decoding matrix $\mathbf{D}$ under which a multi-user linearly separable distributed computing scheme can be equipped with data secrecy. In the theorem we will use $\mathbf{D}(:,\mathrm{Sup}(\mathbf{d}_k^{\mathsf{T}}))$ to denote the submatrix of $\mathbf{D}$ formed by the columns indexed by the support of $\mathbf{d}_k^{\mathsf{T}}$, and we will use
$\mathbf{D} \setminus \mathbf{D}(:,\mathrm{Sup}(\mathbf{d}_k^{\mathsf{T}}))$
to denote the submatrix obtained by removing these columns from $\mathbf{D}$. Finally, we will use $\lambda_{\mathrm{max}}$ and $\lambda_{\mathrm{min}}$ to denote maximum and minimum eigenvalues, respectively, while naturally we maintain our operating assumption that $N\geq K$.

\begin{Theorem}
\label{sufficient_condition}
A multi-user linearly separable distributed computing scheme can provide data secrecy if and only if
\begin{align}
\mathrm{Rank}\!\left(
\mathbf{D} \setminus \mathbf{D}(:,\mathrm{Sup}(\mathbf{d}_k^{\mathsf{T}}))
\right)
\geq K - 1,
\label{rank_condition}
\end{align}
for every user $k$.  For the case of $\mathbf{GF}(q)$, this entails zero information leakage, while in the case of the reals, the information leakage to user $k \in [K]$ is upper bounded as
\begin{align}
I\!\left(
\mathcal{W}; \mathcal{A}_k, \mathbf{f}_k
\,\middle|\,
\langle \mathbf{f}_k, \mathbf{w} \rangle
\right)
\leq
\frac{w_{\mathsf{H}}(\mathbf{d}^{\mathsf{T}}_k) - 1}{2}
\log\!\left(
1 + M_k \frac{\sigma_w^2}{\sigma_c^2}
\right),
\label{real_data_leakage}
\end{align}
where $$M_k
=
\frac{ \lambda_{\mathrm{max}}(\mathbf{X}_k\mathbf{X}_k^{\mathsf{T}})}{\lambda_{\mathrm{min}}(\mathbf{Y}_k\mathbf{Y}_k^{\mathsf{T}})},$$ $\mathbf{X}_k := \mathbf{E}(\mathcal{S}_k,:)$ 
and
$\mathbf{Y}_k := \mathbf{C}(\mathcal{S}_k,:)$, 
with 
$\mathcal{S}_k$
denoting the index set of a maximal linearly independent subset of server responses in matrix $\mathbf{C}(\mathrm{Sup}(\mathbf{d}_k^{\mathsf{T}}),:).$

\end{Theorem}

\begin{IEEEproof}
 The proof contains two parts: the converse proof of~\eqref{rank_condition} is provided in 
Appendix~\ref{proof_ds_finite}. The achievable scheme that induces data secrecy when~\eqref{rank_condition} holds is provided in Section~\ref{Achievable_algorithm}. The proof of the real-valued leakage bound in~\eqref{real_data_leakage} is given in Appendix~\ref{proof_ds_real}.
\end{IEEEproof}

The following remark provides intuition for the bound in~\eqref{real_data_leakage} by relating the information leakage to the per-response information rate and the number of observable server responses.

\begin{Remark}
\label{remark_info_leakage_intuition}
 
The upper bound in~\eqref{real_data_leakage} can be interpreted using concepts from Gaussian communication subchannels. 
Each server response observed by user $k$ is a linear combination of the message vector $\mathbf{w}$ (the \emph{signal}) and the common randomness $\mathbf{c}$ (the \emph{noise}). 
Specifically, $\sigma_w^2 \mathbf{X}_k \mathbf{X}_k^{\mathsf{T}}$ captures the contribution of the messages, while $\sigma_c^2 \mathbf{Y}_k \mathbf{Y}_k^{\mathsf{T}}$ captures the masking effect of the randomness.
 
The ratio $M_k = \lambda_{\mathrm{max}}(\mathbf{X}_k \mathbf{X}_k^{\mathsf{T}}) / \lambda_{\mathrm{min}}(\mathbf{Y}_k \mathbf{Y}_k^{\mathsf{T}})$ can be viewed as a generalized \emph{signal-to-noise ratio (SNR)}, determining how much information about $\mathbf{w}$ can leak through each effective response. 
Meanwhile, $w_{\mathsf{H}}(\mathbf{d}_k^{\mathsf{T}})-1$ corresponds to the number of linearly independent observations user $k$ can access beyond the requested function, i.e., the number of independent subchannels available for leakage.
 
Hence, the total information leakage scales like the product of the number of effective observations and the per-response information rate, analogous to the total capacity of parallel Gaussian subchannels. 
Increasing the variance $\sigma_c^2$ of the common randomness reduces the SNR and therefore suppresses information leakage, providing data secrecy in the real-field setting.
\end{Remark}

\section{Achievable Design}
\label{Achievable_algorithm}
This section introduces the core idea for inducing data secrecy in multi-user linearly separable distributed computing schemes and presents a general achievable design. We first illustrate the approach through a simple example and then describe the general construction. Finally, we characterize classes of decoding matrices $\mathbf{D}$ that satisfy the condition in Theorem~\ref{sufficient_condition}.

\subsection{Example Scheme}
\label{achivable_example}
\begin{Example}
\label{example_sec}
Consider the $(N=6, K=4, L=5)$ multi-user distributed computing scenario as described in Example~\ref{Discrete_Ex}, where recall that the decoding matrix took the form
\[\mathbf{D} =
\begin{bmatrix}
2 & 1 & -1 & 0 & 0 & 0\\
0 & 1 & 1 & 1 & 0 & 0\\
0 & 0 & 1 & 0 & 3 & 0\\
0 & 0 & 0 & 1 & 1 & 1
\end{bmatrix}.\]

The support of the columns directly yields the broadcast sets $\tau_1 = \{1\},\tau_2 = \{1,2\},\tau_3 = \{1,2,3\},\tau_4 = \{2,4\},\tau_5 = \{3,4\},\tau_6 = \{4\},$ each defining the users that a server communicates to, while the supports of the rows yield 
$\mathrm{Sup}(\mathbf{d}^{\mathsf{T}}_1) = \{1,2,3\}, \mathrm{Sup}(\mathbf{d}^{\mathsf{T}}_2) = \{2,3,4\}, \mathrm{Sup}(\mathbf{d}^{\mathsf{T}}_3) = \{3,5\}, \mathrm{Sup}(\mathbf{d}^{\mathsf{T}}_4) = \{4,5,6\}$, each describing the servers that each user collects from.

\paragraph{Step 1 -- Verifying the sufficient condition}
Directly from the above supports $\mathrm{Sup}(\mathbf{d}^{\mathsf{T}}_1),\cdots,\mathrm{Sup}(\mathbf{d}^{\mathsf{T}}_4)$, we first get, for each user $k$, the reduced matrices
$\mathbf{D}_{\mathrm{Red},k} \triangleq \mathbf{D} \setminus \mathbf{D}(:,\mathrm{Sup}(\mathbf{d}^{\mathsf{T}}_k))$, which in our example are given by
\begin{align}
\mathbf{D}_{\mathrm{Red},1} &=
\begin{bmatrix}
0 & 0 & 0\\
1 & 0 & 0\\
0 & 3 & 0\\
1 & 1 & 1
\end{bmatrix}, \quad
\mathbf{D}_{\mathrm{Red},2} =
\begin{bmatrix}
2 & 0 & 0\\
0 & 0 & 0\\
0 & 3 & 0\\
0 & 1 & 1
\end{bmatrix}, \nonumber\\
\mathbf{D}_{\mathrm{Red},3} &=
\begin{bmatrix}
2 & 1 & 0 & 0\\
0 & 1 & 1 & 0\\
0 & 0 & 0 & 0\\
0 & 0 & 1 & 1
\end{bmatrix}, \quad
\mathbf{D}_{\mathrm{Red},4} =
\begin{bmatrix}
2 & 1 & -1\\
0 & 1 & 1\\
0 & 0 & 1\\
0 & 0 & 0
\end{bmatrix}. \nonumber
\end{align}
We can now easily see that for all $k \in [4]$, we have $
\mathrm{Rank}(\mathbf{D}_{\mathrm{Red},k}) = 3 = K-1,$
and we can thus conclude that the condition of Theorem~\ref{sufficient_condition} is satisfied, and secrecy can be induced. Let us now see how, in the following steps, the scheme is shaped in order to provide secrecy. 

\paragraph{Step 2 -- Computing a basis for $\mathrm{Null}(\mathbf{D})$}
It is direct to see that the null space of $\mathbf{D}$ is given by
$\mathrm{Null}(\mathbf{D}) =
\mathrm{Span}\left\{
(-7, 8, -6, -2, 2, 0)^{\mathsf{T}},
(-1, 2, 0, -2, 0, 2)^{\mathsf{T}}
\right\}.$

\paragraph{Step 3 -- Constructing the augmented encoding matrix}
Let us first construct the $N\times (N-K) = 6\times 2$ matrix $\mathbf{C}$ using the above basis vectors of $\mathrm{Null}(\mathbf{D})$ (such that the columns of $\mathbf{C}$ span $\mathrm{Null}(\mathbf{D})$), and let us augment $\mathbf{E}$ to take the form $\tilde{\mathbf{E}} = [\mathbf{E}, \mathbf{C}]$, yielding
\begin{align}
\tilde{\mathbf{E}} =
\begin{bmatrix}
1 & 0 & 0 & 2 & 0 & -7 & -1 \\
1 & 0 & -2 & 3 & 0 & 8 & 2\\
0 & 0 & 1 & 3 & 1 & -6 & 0\\
-1 & 0 & 3 & 0 & 0 & -2 & -2\\
0 & 1 & 0 & 0 & 0 & 2 & 0\\
4 & 2 & 3 & 1 & 1 & 0 & 2
\end{bmatrix}.
\end{align}

\paragraph{Step 4  -- Generating augmented message vector}
Recall that we have $L=5$ function-output datasets $W_1, W_2, \ldots, W_5.$ Let us now create $N-K = 2$ random datasets $C_1$ and $C_2$, which are formed to be independent shared random variables uniformly distributed with the same entropy as the messages. This results to an augmented message vector $$\tilde{\mathbf{w}} = [W_1, W_2, \ldots, W_5, C_1, C_2]^{\mathsf{T}}.$$
 \paragraph{Step 5  -- Generating secured responses}
Each server $n$ computes $A_n = \tilde{\mathbf{e}}_n^{\mathsf{T}} \tilde{\mathbf{w}},$ and transmits the result to all users in $\tau_n.$ 

These server-side computations are summarized in Table~\ref{server_comp_sec}.

\begin{table}[tb]
\vspace{6pt}
\caption{Server-side computations in Example~\ref{example_sec}.}
\centering
\renewcommand{\arraystretch}{1.3}
\begin{tabular}{|c|c|c|}
\hline
Server & Computed linear combination $A_n =\tilde{\mathbf{e}}_n^{\mathsf{T}}\tilde{\mathbf{w}}$ & Broadcast set $\tau_n$ \\ \hline
1 & $W_1 + 2W_4-7C_1-C_2$ & $\{1\}$ \\ \hline
2 & $W_1 - 2W_3 + 3W_4+8C_1+2C_2$ & $\{1,2\}$ \\ \hline
3 & $W_3 + 3W_4 + W_5-6C_1$ & $\{1,2,3\}$ \\ \hline
4 & $-W_1 + 3W_3-2C_1-2C_2$ & $\{2,4\}$ \\ \hline
5 & $W_2+2C_1$ & $\{3,4\}$ \\ \hline
6 & $4W_1 + 2W_2 + 3W_3 + W_4 + W_5+2C_2$ & $\{4\}$ \\ \hline
\end{tabular}
\label{server_comp_sec}
\end{table}
 \paragraph{Step 6  -- Computing requested linear combinations} 
Each user $k \in [4],$ computes  $\mathbf{d}_k^{\mathsf{T}}\mathbf{A}$, with 
$$\mathbf{A} = [A_1, \ldots, A_6]^{\mathsf{T}},$$ 
to recover its requested linear combination. These user-side computations are similar to Example~\ref{Discrete_Ex}, and are summarized in Table~\ref{user_comp}.
\end{Example}

\paragraph{Verifying conditions in example}
We now show that the correctness condition in~\eqref{correctness_condition} and the data secrecy condition~\eqref{data_secrecy_condition}, both hold.

First, the correctness condition in~\eqref{correctness_condition} is easily shown to hold, since by design we have $\mathbf{D}\mathbf{E} = \mathbf{F}$ and we have that each column of $\mathbf{C}$ lies in $\mathrm{Null}(\mathbf{D})$, which automatically yields
$$\mathbf{D}\tilde{\mathbf{E}} = \mathbf{D}[\mathbf{E}, \mathbf{C}] = [\mathbf{F}, \mathbf{0}].$$

 To verify data secrecy, for each user $k$, we first note that the received responses are
\begin{align}
\mathcal{A}_1 = \{&W_1+2W_4-7C_1-C_2,\;\nonumber\\ &W_1-2W_3+3W_4+8C_1+2C_2,\;\nonumber\\ &W_3+3W_4+W_5-6C_1\},\nonumber\\
\mathcal{A}_2 = \{&W_1-2W_3+3W_4+8C_1+2C_2,\;\nonumber\\ &W_3+3W_4+W_5-6C_1,\;\nonumber\\
&-W_1+3W_3-2C_1-2C_2\},\nonumber\\
\mathcal{A}_3 = \{&W_3+3W_4+W_5-6C_1,\;\nonumber\\&W_2+2C_1\},\nonumber\\
\mathcal{A}_4 = \{&-W_1+3W_3-2C_1-2C_2,\; \nonumber\\&W_2+2C_1,\; \nonumber\\&4W_1+2W_2+3W_3+W_4+W_5+2C_2\}.\nonumber
\end{align}
We now prove data secrecy by contradiction. Suppose user $k \in [4]$ can recover two linearly independent linear combinations of the messages. This would mean that there exist two linearly independent decoding vectors $\mathbf{a}_k^{\mathsf{T}}$ and $\tilde{\mathbf{a}}_k^{\mathsf{T}}$ such that both cancel the injected randomness 
\[
\mathbf{a}_k^{\mathsf{T}}\mathbf{C}=\mathbf{0}^{\mathsf{T}},
\tilde{\mathbf{a}}_k^{\mathsf{T}}\mathbf{C}=\mathbf{0}^{\mathsf{T}},
\]
while yielding nonzero message combinations.
Since decoding is performed only from servers accessible to the user $k$, both vectors are supported on $\mathrm{Sup}(\mathbf{d}_k^{\mathsf{T}})$:
\[
\mathrm{Sup}(\mathbf{a}_k^{\mathsf{T}}),\;
\mathrm{Sup}(\tilde{\mathbf{a}}_k^{\mathsf{T}})
\subseteq \mathrm{Sup}(\mathbf{d}_k^{\mathsf{T}}),
\]
which would mean that the submatrix $\mathbf{C}(\mathrm{Sup}(\mathbf{d}_k^{\mathsf{T}}),:)$ has at least a two-dimensional left null space, in turn implying that
\begin{align}
\label{rank_example}
    \mathrm{Rank}\!\big(\mathbf{C}(\mathrm{Sup}(\mathbf{d}_k^{\mathsf{T}}),:)\big)
\leq w_{\mathsf{H}}(\mathbf{d}_k^{\mathsf{T}}) - 2.
\end{align}
Now, for the given construction, the relevant submatrices are
\begin{align}
\mathbf{C}(\mathrm{Sup}(\mathbf{d}^{\mathsf{T}}_1),:) &=
\begin{bmatrix}
-7 & -1\\
8 & 2\\
-6 & 0
\end{bmatrix}, &
\mathbf{C}(\mathrm{Sup}(\mathbf{d}^{\mathsf{T}}_2),:) &=
\begin{bmatrix}
8 & 2\\
-6 & 0\\
-2 & -2
\end{bmatrix}, \nonumber\\
\mathbf{C}(\mathrm{Sup}(\mathbf{d}^{\mathsf{T}}_3),:) &=
\begin{bmatrix}
-6 & 0\\
2 & 0
\end{bmatrix}, &
\mathbf{C}(\mathrm{Sup}(\mathbf{d}^{\mathsf{T}}_4),:) &=
\begin{bmatrix}
-2 & 2\\
2 & 0\\
0 & 2
\end{bmatrix}. \nonumber
\end{align}
A direct rank calculation shows that for all $k \in [4]$,
\[
\mathrm{Rank}\!\big(\mathbf{C}(\mathrm{Sup}(\mathbf{d}_k^{\mathsf{T}}),:)\big)
= w_{\mathsf{H}}(\mathbf{d}_k^{\mathsf{T}}) - 1,
\]
which contradicts the inequality~\eqref{rank_example}. Therefore, no user can form two linearly independent decoding vectors that eliminate the randomness. Thus, the data secrecy condition in~\eqref{data_secrecy_condition} holds, and each user can recover only its requested linear combination and no additional information about the messages.
\paragraph{Information leakage over $\mathbb{R}$}
Assuming now real-valued computations, for each $l\in[5]$, let
$W_l \sim \mathcal{N}(0,\sigma_w^2),$
$C_j \sim \mathcal{N}(0,\sigma_c^2), \; j\in[2],$
all being mutually independent.

 Focusing on user 1, using~\eqref{query_indepence}, we can see that
\begin{align}
I\!\left(\mathcal{W};\mathcal{A}_1,\mathbf{f}_1 \mid \langle \mathbf{f}_1,\mathbf{w}\rangle\right)
&= I\!\left(\mathcal{W};\mathcal{A}_1 \mid \mathbf{f}_1,\langle \mathbf{f}_1,\mathbf{w}\rangle\right) \nonumber\\
&= h(\mathcal{A}_1 \mid \mathbf{f}_1,\langle \mathbf{f}_1,\mathbf{w}\rangle)
 - h(\mathcal{A}_1 \mid \mathbf{f}_1,\mathcal{W}).
 \label{information_leakage_ex}
\end{align}
We now compute the two terms in~\eqref{information_leakage_ex}.

 \textbf{Step 1: Computing} $h(\mathcal{A}_1 \mid \mathbf{f}_1,\langle \mathbf{f}_1,\mathbf{w}\rangle)$.
 
By the chain rule,
\begin{align}
h(\mathcal{A}_1 \mid \mathbf{f}_1,\langle \mathbf{f}_1,\mathbf{w}\rangle)
&= h(A_2, A_3 \mid \mathbf{f}_1,\langle \mathbf{f}_1,\mathbf{w}\rangle) \nonumber\\
&+ h(A_1 \mid A_2,A_3,\mathbf{f}_1,\langle \mathbf{f}_1,\mathbf{w}\rangle).
\end{align}
By correctness~\eqref{correctness_condition}, $A_1$ is a deterministic function of $(A_2,A_3,\langle \mathbf{f}_1,\mathbf{w}\rangle)$, hence
\[
h(A_1 \mid A_2,A_3,\mathbf{f}_1,\langle \mathbf{f}_1,\mathbf{w}\rangle)=0.
\]
Since conditioning does not increase entropy,
\[
h(\mathcal{A}_1 \mid \mathbf{f}_1,\langle \mathbf{f}_1,\mathbf{w}\rangle)
= h(A_2, A_3 \mid \mathbf{f}_1,\langle \mathbf{f}_1,\mathbf{w}\rangle) \le h(A_2, A_3).
\]

Note that $A_3 = W_3 + 3W_4 + W_5 - 6C_1, A_2 = W_1 - 2W_3 + 3W_4 + 8C_1 + 2C_2.$ Let us define the matrices $\mathbf{X}_1$ and $\mathbf{Y}_1$ as below,
 $$\mathbf{X}_1=
 \begin{bmatrix}
1 & 0 &-2 & 3 & 0\\
0 & 0 &1 & 3 & 1
\end{bmatrix},
\mathbf{Y}_1=
 \begin{bmatrix}
8 & 2 \\
-6 & 0
\end{bmatrix}.$$
Then
\begin{align}
    h(A_2, A_3) = h(\mathbf{X}_1\mathbf{w}+\mathbf{Y}_1\mathbf{c}),
\end{align}
where $\mathbf{w}=[W_1, W_2,\ldots, W_5]^{\mathsf{T}}$
and $\mathbf{c}=[C_1, C_2]^{\mathsf{T}}$.
Since $\mathbf{w}$ and $\mathbf{c}$ are independent, we have
\begin{align}
    \Sigma_{\mathbf{X}_1\mathbf{w}+\mathbf{Y}_1\mathbf{c}}=\mathrm{Cov}(\mathbf{X}_1\mathbf{w}+\mathbf{Y}_1\mathbf{c}) = \sigma^2_w\mathbf{X}_1\mathbf{X}_1^{\mathsf{T}}+ \sigma^2_c\mathbf{Y}_1\mathbf{Y}_1^{\mathsf{T}}.
\end{align}
Since $\mathbf{X}_1\mathbf{w}+\mathbf{Y}_1\mathbf{c}$ is a Gaussian random vector, we have
\begin{align}
    h(\mathbf{X}_1\mathbf{w}+\mathbf{Y}_1\mathbf{c}) &= \frac{1}{2}\log((2\pi e)^2|\Sigma_{\mathbf{X}_1\mathbf{w}+\mathbf{Y}_1\mathbf{c}}|)\nonumber\\
    &= \frac{1}{2}\log((2\pi e)^2|\sigma^2_w\mathbf{X}_1\mathbf{X}_1^{\mathsf{T}}+ \sigma^2_c\mathbf{Y}_1\mathbf{Y}_1^{\mathsf{T}}|),
\end{align}
where $|\Sigma_{\mathbf{X}_1\mathbf{w}+\mathbf{Y}_1\mathbf{c}}|$ is the determinant of matrix $\Sigma_{\mathbf{X}_1\mathbf{w}+\mathbf{Y}_1\mathbf{c}}.$
\textbf{Step 2: Computing} $h(\mathcal{A}_1 \mid \mathbf{f}_1,\mathcal{W}).$ 
\begin{align}
    h(\mathcal{A}_1 \mid \mathbf{f}_1,\mathcal{W})&=h(-7C_1-C_2, 8C_1+2C_2, -6C_1).
\end{align}

 Since
\(
\mathrm{Rank}\big(\mathbf{C}(\mathrm{Sup}(\mathbf{d}_1^{\mathsf{T}}),:)\big)=2,
\)
these three variables span a 2-dimensional Gaussian subspace. Thus,
\begin{align}
h(\mathcal{A}_1 \mid \mathbf{f}_1,\mathcal{W})
&= h(8C_1+2C_2,\,-6C_1)\nonumber\\
&= h(\mathbf{Y_1c}) = \frac{1}{2}\log((2\pi e)^2|\sigma^2_c\mathbf{Y}_1\mathbf{Y}_1^{\mathsf{T}}|).
\end{align}

\textbf{Step 3: Combining Steps 1 and 2.}
Combining Steps~1 and~2 yields
\begin{align}
    I\!\left(\mathcal{W};\mathcal{A}_1,\mathbf{f}_1 \mid \langle \mathbf{f}_1,\mathbf{w}\rangle\right)
&\leq \frac{1}{2}\log((2\pi e)^2|\sigma^2_w\mathbf{X}_1\mathbf{X}_1^{\mathsf{T}}+ \sigma^2_c\mathbf{Y}_1\mathbf{Y}_1^{\mathsf{T}}|)\nonumber\\
&- \frac{1}{2}\log((2\pi e)^2|\sigma^2_c\mathbf{Y}_1\mathbf{Y}_1^{\mathsf{T}}|)\nonumber\\
&=\frac{1}{2}\log(\frac{|\sigma^2_w\mathbf{X}_1\mathbf{X}_1^{\mathsf{T}}+ \sigma^2_c\mathbf{Y}_1\mathbf{Y}_1^{\mathsf{T}}|}{| \sigma^2_c\mathbf{Y}_1\mathbf{Y}_1^{\mathsf{T}}|}).
\end{align}
Since $\mathbf{Y}_1$ is full rank and invertible, we have
\begin{align}
    \frac{1}{2}\log(&\frac{|\sigma^2_w\mathbf{X}_1\mathbf{X}_1^{\mathsf{T}}+ \sigma^2_c\mathbf{Y}_1\mathbf{Y}_1^{\mathsf{T}}|}{| \sigma^2_c\mathbf{Y}_1\mathbf{Y}_1^{\mathsf{T}}|})\nonumber\\
    &=\frac{1}{2}\log(|\mathbf{I}_2+\frac{\sigma_w^2}{\sigma_c^2}\mathbf{X}_1\mathbf{X}_1^{\mathsf{T}}(\mathbf{Y}_1\mathbf{Y}_1^{\mathsf{T}})^{-1}|).
\end{align}
However, we have
\begin{align}
    |\mathbf{I}_2+\frac{\sigma_w^2}{\sigma_c^2}\mathbf{X}_1&\mathbf{X}_1^{\mathsf{T}}(\mathbf{Y}_1\mathbf{Y}_1^{\mathsf{T}})^{-1}|\nonumber\\
    &=\prod_{i=1}^2(1+\lambda_i(\frac{\sigma_w^2}{\sigma_c^2}\mathbf{X}_1\mathbf{X}_1^{\mathsf{T}}(\mathbf{Y}_1\mathbf{Y}_1^{\mathsf{T}})^{-1}))\nonumber\\
    &\leq (1+\lambda_{\mathrm{max}}(\frac{\sigma_w^2}{\sigma_c^2}\mathbf{X}_1\mathbf{X}_1^{\mathsf{T}}(\mathbf{Y}_1\mathbf{Y}_1^{\mathsf{T}})^{-1}))^2.
\end{align}
where for each $i \in [2]$, $\lambda_i(\frac{\sigma_w^2}{\sigma_c^2}\mathbf{X}_1\mathbf{X}_1^{\mathsf{T}}(\mathbf{Y}_1\mathbf{Y}_1^{\mathsf{T}})^{-1})$ is the eigenvalue of matrix $\frac{\sigma_w^2}{\sigma_c^2}\mathbf{X}_1\mathbf{X}_1^{\mathsf{T}}(\mathbf{Y}_1\mathbf{Y}_1^{\mathsf{T}})^{-1},$ and $\lambda_{\mathrm{max}}(\frac{\sigma_w^2}{\sigma_c^2}\mathbf{X}_1\mathbf{X}_1^{\mathsf{T}}(\mathbf{Y}_1\mathbf{Y}_1^{\mathsf{T}})^{-1}))$ is the maximum eigenvalue.
Since $\sigma^2_w\mathbf{X}_1\mathbf{X}_1^{\mathsf{T}}$ and $ \sigma^2_c\mathbf{Y}_1\mathbf{Y}_1^{\mathsf{T}}$ are positive semi-definite, we have
\begin{align}
    \lambda_{\mathrm{max}}(\frac{\sigma_w^2}{\sigma_c^2}\mathbf{X}_1\mathbf{X}_1^{\mathsf{T}}(\mathbf{Y}_1\mathbf{Y}_1^{\mathsf{T}})^{-1})\leq \frac{ \lambda_{\mathrm{max}}(\sigma_w^2\mathbf{X}_1\mathbf{X}_1^{\mathsf{T}})}{\lambda_{\mathrm{min}}(\sigma_c^2\mathbf{Y}_1\mathbf{Y}_1^{\mathsf{T}})}.
\end{align}
Which leads
\begin{align}
    I\!\left(\mathcal{W};\mathcal{A}_1,\mathbf{f}_1 \mid \langle \mathbf{f}_1,\mathbf{w}\rangle\right)
&\leq \frac{2}{2}\log(1+\frac{\sigma_w^2 \lambda_{\mathrm{max}}(\mathbf{X}_1\mathbf{X}_1^{\mathsf{T}})}{\sigma_c^2\lambda_{\mathrm{min}}(\mathbf{Y}_1\mathbf{Y}_1^{\mathsf{T}})})\nonumber\\
&= \log(1+\frac{19.66\sigma_w^2 }{1.4\sigma_c^2}),
\end{align}
Where $\lambda_{\mathrm{max}}(\mathbf{X}_1\mathbf{X}_1^{\mathsf{T}})=19.66$ and $\lambda_{\mathrm{min}}(\mathbf{Y}_1\mathbf{Y}_1^{\mathsf{T}})=1.4$. Similarly for user $k\in[4]$, we can define $\mathbf{X}_k$ and $\mathbf{Y}_k$, such that,
\begin{align}
    &I\!\left(\mathcal{W};\mathcal{A}_k,\mathbf{f}_k \mid \langle \mathbf{f}_k,\mathbf{w}\rangle\right)\nonumber\\
    &\leq \frac{w_{\mathsf{H}}(\mathbf{d}_k^{\mathsf{T}})-1}{2}\log(1+\frac{\sigma_w^2 \lambda_{\mathrm{max}}(\mathbf{X}_k\mathbf{X}_k^{\mathsf{T}})}{\sigma_c^2\lambda_{\mathrm{min}}(\mathbf{Y}_k\mathbf{Y}_k^{\mathsf{T}})}).
\end{align}
so for any $\varepsilon>0$, choosing $\sigma_c^2$ sufficiently large ensures
\[
I\!\left(
\mathcal{W}; \mathcal{A}_k, \mathbf{f}_k
\,\middle|\,
\langle \mathbf{f}_k, \mathbf{w} \rangle
\right) \le \varepsilon,
\quad \forall k\in[4],
\]
which establishes~\eqref{Privacy_con}.

Finally, the communication and computation costs remain the same as in Example~\ref{Discrete_Ex}. In particular, $\gamma = \frac{2}{3}$ and $\delta = \frac{11}{24} \le 1 - \frac{K-1}{N} = \frac{1}{2},$
as required by Theorem~\ref{Theorem_necessary_condition}. Note that the computation cost \(\gamma\) is defined based on the computation tasks assigned to the servers and depends on \(\mathbf{E}\), not on the augmented matrix \(\tilde{\mathbf{E}}\). Therefore, the pair \((\gamma,\delta)\) is feasible for the proposed scheme.

Based on the core ideas introduced in Section~\ref{achivable_example}, we now describe the general scheme.  

\subsection{General Scheme}
Let us now describe the scheme in general terms.  
We start by assuming a known request matrix \(\mathbf{F} = [\mathbf{f}_1, \mathbf{f}_2, \ldots, \mathbf{f}_K]^{\mathsf{T}} \in \mathbb{F}^{K \times L},\) a known existing (potentially non-secure) scheme defined by the encoding matrix $\mathbf{E} = [\mathbf{e}_1, \mathbf{e}_2, \ldots, \mathbf{e}_N]^{\mathsf{T}} \in \mathbb{F}^{N \times L}$ and the decoding matrix $\mathbf{D} = [\mathbf{d}_1, \mathbf{d}_2, \ldots, \mathbf{d}_K]^{\mathsf{T}} \in \mathbb{F}^{K \times N}$. We also assume that each server $n$ has computed the outputs (from the set $\{W_1, \ldots, W_L\}$) as this is defined by the support of their corresponding row $\mathbf{e}_n$ of $\mathbf{E}$. We now focus on describing the procedure that transforms any admissible scheme into a scheme with data secrecy.

\textbf{Step 1:}  
Given the above decomposition $\mathbf{F} = \mathbf{D} \mathbf{E},$ the master node checks whether $\mathbf{D}$ satisfies the rank condition in~\eqref{rank_condition}. If the condition is not satisfied, a different factorization is chosen.

\textbf{Step 2:}  
The master node computes a basis for $\mathrm{Null}(\mathbf{D}).$ 

\textbf{Step 3:}  
The master node arranges the basis vectors of $\mathrm{Null}(\mathbf{D})$ as columns of a matrix $\mathbf{C}$ and constructs the augmented encoding matrix $$\tilde{\mathbf{E}} = [\mathbf{E}, \mathbf{C}]
= [\tilde{\mathbf{e}}_1, \tilde{\mathbf{e}}_2, \ldots, \tilde{\mathbf{e}}_N]^{\mathsf{T}}.
$$

\textbf{Step 4:} 
Using shared common randomness $\mathcal{C}$, the servers generate the independent random datasets
$C_1, C_2, \ldots, C_{N-K}$, which are common across the servers.  In the case of operating over $\mathbb{R}$, these are distributed as $C_i \sim \mathcal{N}(0,\sigma_c^2)$ for some sufficiently large $\sigma^2_c$, and in the case of $\mathbf{GF}(q)$, these are chosen with uniform distribution to match the entropy of the original output messages.  The augmented message vector 
$$\tilde{\mathbf{w}}
= [W_1, \ldots, W_L, C_1, \ldots, C_{N-K}]^{\mathsf{T}}$$
is formed. 

\textbf{Step 5:}  
Each server $n \in [N]$ computes
\(
A_n = \tilde{\mathbf{e}}_n^{\mathsf{T}} \tilde{\mathbf{w}},
\)
and transmits the result to all users in $\tau_n$, where we recall that each $\tau_n$ is defined by the support of the $n$-th column of $\mathbf{D}$.  We also recall that, by design, the support of each $\tilde{\mathbf{e}}_n$ defines which of the output datasets $\{W_1, \ldots, W_L\}$ and the random dataset $\{C_1, \ldots, C_{N-K}\}$ are needed at each server $n$ in order to generate $A_n$. 

\textbf{Step 6:} 
Each user $k \in [K]$ computes $\mathbf{d}_k^{\mathsf{T}}\mathbf{A}$, with 
$$\mathbf{A} = [A_1, \ldots, A_N]^{\mathsf{T}},$$ 
to recover its requested linear combination. Again, let us recall that the available and required elements of $\mathbf{A}$ that are needed by user $k$ are defined by the support $\mathrm{Sup}(\mathbf{d}^{\mathsf{T}}_k)$.

To complete the achievability proof of Theorem~\ref{sufficient_condition}, we now show that this scheme satisfies the correctness condition~\eqref{correctness_condition} and the data secrecy conditions~\eqref{data_secrecy_condition} and~\eqref{Privacy_con}.

\subsection{Achievability Proof of 
Theorem~\ref{sufficient_condition}} 
\subsubsection{Verifying correctness}
From Steps 1–3, $$\mathbf{D}\tilde{\mathbf{E}}
= \mathbf{D}[\mathbf{E}, \mathbf{C}]
= [\mathbf{F}, \mathbf{0}],$$
since $\mathbf{D}\mathbf{E} = \mathbf{F}$ and each column of $\mathbf{C}$ lies in $\mathrm{Null}(\mathbf{D})$. Hence, the correctness condition in~\eqref{correctness_condition} holds. 

\subsubsection{Verifying data secrecy} We proceed by contradiction. Suppose that there exists a user $k \in [K]$ who can recover two linearly independent message-bearing linear combinations $\langle \mathbf{f}_k, \mathbf{w} \rangle$ and $\langle \tilde{\mathbf{f}}_k, \mathbf{w} \rangle.$ Since  the submatrix $\mathbf{E(\mathrm{Sup(\mathbf{d}_k^{\mathsf{T}}),:})}$ is full rank, any decoding vector that cancels the common randomness necessarily produces distinct nonzero linear combination of the messages.
So, there exist two linearly independent decoding vectors
$\mathbf{a}_k^{\mathsf{T}}$ and $\tilde{\mathbf{a}}_k^{\mathsf{T}}$ that cancel the randomness and yield nonzero linear combinations of the messages, i.e.,
$$\langle \mathbf{f}_k, \mathbf{w} \rangle=\mathbf{a}_k^{\mathsf{T}}\mathbf{Ew}\neq\tilde{\mathbf{a}}_k^{\mathsf{T}}\mathbf{Ew}=\langle \tilde{\mathbf{f}}_k, \mathbf{w} \rangle,$$
and
$$\mathbf{a}_k^{\mathsf{T}}\mathbf{Cc}=\tilde{\mathbf{a}}_k^{\mathsf{T}}\mathbf{Cc}=0.$$ 

By construction, $\mathbf{C}$ has full-rank and both decoding vectors
$\mathbf{a}_k^{\mathsf{T}}$ and $\tilde{\mathbf{a}}_k^{\mathsf{T}}$ belong to the row span of the decoding matrix, i.e., $$\mathbf{a}_k^{\mathsf{T}}, \tilde{\mathbf{a}}_k^{\mathsf{T}} \in 
\mathrm{Span}\{\mathbf{d}_1^{\mathsf{T}}, \mathbf{d}_2^{\mathsf{T}}, \ldots, \mathbf{d}_K^{\mathsf{T}}\}.$$
Hence, there exist linearly independent vectors
$\mathbf{b}_k, \tilde{\mathbf{b}}_k \in \mathbb{F}^K$
such that $$\mathbf{a}_k^{\mathsf{T}} = \mathbf{b}_k^{\mathsf{T}} \mathbf{D}, \tilde{\mathbf{a}}_k^{\mathsf{T}} = \tilde{\mathbf{b}}_k^{\mathsf{T}} \mathbf{D}.$$
Moreover, since both decoding vectors are supported only on the servers accessible to the user~$k$, we have $$\mathrm{Sup}(\mathbf{a}_k^{\mathsf{T}}), \mathrm{Sup}(\tilde{\mathbf{a}}_k^{\mathsf{T}})
\subseteq \mathrm{Sup}(\mathbf{d}^{\mathsf{T}}_k).$$ For $$\mathbf{D}_{\mathrm{Red},k}
= \mathbf{D} \setminus \mathbf{D}(:, \mathrm{Sup}(\mathbf{d}^{\mathsf{T}}_k)),$$ obtained by removing the columns of $\mathbf{D}$ indexed by $\mathrm{Sup}(\mathbf{d}^{\mathsf{T}}_k)$, we directly have that $$\mathbf{b}_k^{\mathsf{T}} \mathbf{D}_{\mathrm{Red},k} = \mathbf{0}^{\mathsf{T}}, \tilde{\mathbf{b}}_k^{\mathsf{T}} \mathbf{D}_{\mathrm{Red},k} = \mathbf{0}^{\mathsf{T}}.$$
Since $\mathbf{b}_k^{\mathsf{T}}$ and $\tilde{\mathbf{b}}_k^{\mathsf{T}}$ are linearly independent, the null space of $\mathbf{D}_{\mathrm{Red},k}$ has dimension at least two, implying 
\begin{align}
\mathrm{Rank}\bigl(\mathbf{D}_{\mathrm{Red},k}\bigr) \le K - 2,
\label{achiv_proof_Theorem1}
\end{align}
contradicts the rank condition in~\eqref{rank_condition}, which is verified in Step~1 of the construction. Therefore, no user can recover more than a single linear combination of the messages, and the data secrecy requirement in~\eqref{data_secrecy_condition} is satisfied. When computations are performed over $\mathbb{R}$, the information-theoretic leakage bound in~\eqref{real_data_leakage} applies, as proved in Appendix~\ref{proof_ds_real}.

This completes the achievable proof of Theorem~\ref{sufficient_condition}.

\subsection{ Some forms of the decoding matrix $\mathbf{D}$}
To induce data secrecy, the decoding matrix $\mathbf{D}$ must satisfy the condition stated in Theorem~\ref{sufficient_condition}. Several non-trivial classes of matrices naturally satisfy this condition. We highlight a few representative examples. 
\begin{itemize}
\item  \textbf{Systematic form $\mathbf{D} = [\mathbf{I}_K\mid\mathbf{P}]$:}  
    Here, $\mathbf{D}$ is a $K \times N$ matrix, where $\mathbf{I}_K$ denotes the $K \times K$ identity matrix. In $\mathbf{GF}(q)$, $\mathbf{D}$ can be interpreted as the parity-check matrix of a systematic linear code. As before, each column $\mathbf{E}(:,i)$ may then be viewed as a lowest-weight coset leader corresponding to the syndrome $\mathbf{F}(:,i)$~\cite{khalesi2023multi}, corresponding to the same code. Interestingly, the aforementioned basis vectors of $\mathrm{Null}(\mathbf{D})$ naturally form the generator matrix of the linear code whose parity-check matrix is $\mathbf{D}$. Consequently, the common randomness matrix can be chosen as
    $\mathbf{C} = [-\mathbf{P}^{\mathsf{T}} \mid \mathbf{I}_{N-K}]^{\mathsf{T}}.$ This finally tells us that any decoding matrix $\mathbf{D}$ can be transformed to its systematic form (with a possible change in $\gamma,\delta$) in order to guarantee the secrecy structural condition of Theorem~\ref{sufficient_condition}.
\item \textbf{Cyclic code structure:}  
    The matrix $\mathbf{D}$ is the generator matrix of a cyclic code in circulant form. Thus, for each user $k \in [K]$, 
$\mathbf{D} \setminus \mathbf{D}(:, \mathrm{Sup}(\mathbf{d}^{\mathsf{T}}_k))$ 
forms an upper trapezoidal matrix with $K-1$ nonzero diagonal entries. 
Hence, $\mathrm{Rank}\Bigl(\mathbf{D} \setminus \mathbf{D}(:, \mathrm{Sup}(\mathbf{d}^{\mathsf{T}}_k))\Bigr) \ge K-1,$
satisfying the condition stated in Theorem~\ref{sufficient_condition}.
In this case, the corresponding parity-check matrix is also circulant.
\item \textbf{Identity matrix $\mathbf{D}=\mathbf{I}_K$:}  
    Here, $\mathbf{D}$ is a $K \times K$ identity matrix. Each user directly receives their requested linear combination from a single assigned server, but this naturally requires maximal $\gamma = 1$. 
\end{itemize}

\section{Conclusions}
\label{section_conclusions}
In this paper, we have studied the problem of secure multi-user linearly-separable distributed computing, providing an information-theoretic framework to guarantee that each user can recover only its requested linear combination while learning nothing else. We established a necessary and sufficient condition on the decoding process, which ensures data secrecy over both finite fields $\mathbf{GF}(q)$ and real numbers $\mathbb{R}$, and presented an achievable scheme that preserves correctness, as well as communication and computation costs. 

Our work is motivated by the observation that many powerful distributed computing frameworks, including those of the multi-user linearly-separable nature, do not inherently provide data secrecy guarantees, and may thus leak unintended information through aggregation or side information. Such breach of data secrecy thus brings to the fore the urgent need for novel designs that maintain a good degree of performance, while guaranteeing various degrees of secrecy. This work provides the first such approach in the multi-user linearly-separable distributed computing setting, and does so by characterizing the precise algebraic and rank-based requirements for secrecy, thus laying foundations for secure designs in large-scale multi-user distributed analytics and computation.  

\section*{Acknowledgment}
The authors would like to thank Mojtaba Tefagh for careful proofreading of the manuscript and for insightful discussions that helped improve the clarity of this work.
\begin{appendices}
\section{Proof of Lemma~\ref{rank_lemma}}
\label{proof_rank_user}
Suppose that user $k$ is connected to $w_{\mathsf{H}}(\mathbf{d}_k^{\mathsf{T}})$ servers, and
\[
\mathrm{Rank}\bigl(\mathbf{C}(\mathrm{Sup}(\mathbf{d}_k^{\mathsf{T}}), :)\bigr)
< w_{\mathsf{H}}(\mathbf{d}_k^{\mathsf{T}}) - 1.
\]
Since $(\mathbf{C}(\mathrm{Sup}(\mathbf{d}_k^{\mathsf{T}}), :))^{\mathsf{T}} \in \mathbb{F}^{(N-K)\times w_{\mathsf{H}}(\mathbf{d}_k^{\mathsf{T}})}$, the rank--nullity theorem implies
\begin{align}
\dim\!\left(\mathrm{Null}\bigl((\mathbf{C}(\mathrm{Sup}(\mathbf{d}_k^{\mathsf{T}}), :))^{\mathsf{T}}\bigr)\right)
&= w_{\mathsf{H}}(\mathbf{d}_k^{\mathsf{T}})\nonumber\\
&- \mathrm{Rank}\bigl(\mathbf{C}(\mathrm{Sup}(\mathbf{d}_k^{\mathsf{T}}), :)\bigr)\nonumber\\
&> 1.
\end{align}

Consequently, there exist at least two non-zero linearly independent vectors 
$\mathbf{a}_k^{\mathsf{T}}$ and $\tilde{\mathbf{a}}_k^{\mathsf{T}}$ supported on $\mathrm{Sup}(\mathbf{d}^{\mathsf{T}}_k),$  $$\mathrm{Sup}(\mathbf{a}_k^{\mathsf{T}}), \mathrm{Sup}(\tilde{\mathbf{a}}_k^{\mathsf{T}})
\subseteq \mathrm{Sup}(\mathbf{d}^{\mathsf{T}}_k),$$
such that
$$(\mathbf{C}(\mathrm{Sup}(\mathbf{d}^{\mathsf{T}}_k), :))^{\mathsf{T}}\mathbf{a}_k=(\mathbf{C}(\mathrm{Sup}(\mathbf{d}^{\mathsf{T}}_k), :))^{\mathsf{T}}\tilde{\mathbf{a}}_k=\mathbf{0}.$$ 
or equivalently,
$$\mathbf{a}_k^{\mathsf{T}}\mathbf{C}(\mathrm{Sup}(\mathbf{d}^{\mathsf{T}}_k), :) =\tilde{\mathbf{a}}_k^{\mathsf{T}}\mathbf{C}(\mathrm{Sup}(\mathbf{d}^{\mathsf{T}}_k), :)=\mathbf{0}^{\mathsf{T}}.$$ 
Since the submatrix $\mathbf{E(\mathrm{Sup(\mathbf{d}_k^{\mathsf{T}}),:})}$ is full rank, any decoding vector that cancels the common randomness necessarily produces distinct nonzero linear combination of the messages, so
\begin{align}
\langle \mathbf{f}_k, \mathbf{w} \rangle
&=\mathbf{a}_k^{\mathsf{T}}\mathbf{E(\mathrm{Sup(\mathbf{d}_k^{\mathsf{T}}),:})}\mathbf{w}\nonumber\\
&\neq\tilde{\mathbf{a}}_k^{\mathsf{T}}\mathbf{E(\mathrm{Sup(\mathbf{d}_k^{\mathsf{T}}),:})}\mathbf{w}=\langle \tilde{\mathbf{f}}_k, \mathbf{w} \rangle.
\end{align}

Therefore, after eliminating the randomness, user~$k$ can recover a message-bearing subspace of dimension greater than one. This implies that user~$k$ can decode more than a single linear combination of messages, and data secrecy cannot be guaranteed for user $k.$

\section{Proof of Theorem~\ref{Theorem_necessary_condition}}
\label{proof_delta_bound}
With Lemma~\ref{rank_lemma} in place, we can derive the following converse bound, which limits the number of servers accessible to each user over which data secrecy is impossible.

\begin{Corollary}
\label{Remark_access_upper}
In a multi-user linearly separable distributed computing scheme, if
\begin{align}
w_{\mathsf{H}}(\mathbf{d}_k^{\mathsf{T}}) > N - K + 1,
\label{access_requirement}
\end{align}
user $k$ can then obtain information beyond its requested linear combination.
\begin{IEEEproof}
We prove the statement by contradiction. Suppose that data secrecy holds for the user $k$. Then, by Lemma~\ref{rank_lemma},
\begin{align}
    w_{\mathsf{H}}(\mathbf{d}_k^{\mathsf{T}})-1
\leq \mathrm{Rank}\!\big(\mathbf{C}(\mathrm{Sup}(\mathbf{d}_k^{\mathsf{T}}),:)\big) .
\end{align}
Under assumption~\eqref{access_requirement}, this implies
\begin{align}
\label{Eq_app_1}
N-K<\mathrm{Rank}\!\big(\mathbf{C}(\mathrm{Sup}(\mathbf{d}_k^{\mathsf{T}}),:)\big).
\end{align}
However,
\begin{align}
\label{Eq_app_2}
\mathrm{Rank}\!\big(\mathbf{C}(\mathrm{Sup}(\mathbf{d}_k^{\mathsf{T}}),:)\big)
\leq \mathrm{Rank}(\mathbf{C}).
\end{align}
Moreover, the correctness condition~\eqref{correctness_condition} requires $\mathbf{D}\mathbf{C} = \mathbf{0}$,
which implies that the column space of $\mathbf{C}$ lies in $\mathrm{Null}(\mathbf{D})$. Since $\mathbf{D}$ has rank $K$, it follows that
\begin{align}
\label{Eq_app_3}
\mathrm{Rank}(\mathbf{C}) \leq N - K.
\end{align}
Combining~\eqref{Eq_app_1}, \eqref{Eq_app_2}, and~\eqref{Eq_app_3}, we obtain
\[
N - K < \mathrm{Rank}\!\big(\mathbf{C}(\mathrm{Sup}(\mathbf{d}_k^{\mathsf{T}}),:)\big)
\leq \mathrm{Rank}(\mathbf{C}) \leq N - K,
\]
which is a contradiction. Hence, data secrecy cannot hold, and the user $k$ obtains information beyond its requested linear combination.
\end{IEEEproof}
\end{Corollary}

Based on Corollary~\ref{Remark_access_upper}, data secrecy for each user $k \in [K]$ requires
\[
w_{\mathsf{H}}(\mathbf{d}^{\mathsf{T}}_k) \leq N - K + 1.
\]
Summing over all users yields
\[
\sum_{k=1}^K w_{\mathsf{H}}(\mathbf{d}^{\mathsf{T}}_k)
\leq K(N - K + 1)
= NK - K^2 + K.
\]
Dividing both sides by $KN$ and using the definition of communication cost in~\eqref{delta_def}, we obtain
\[
\delta
= \frac{1}{KN} \sum_{k=1}^K w_{\mathsf{H}}(\mathbf{d}^{\mathsf{T}}_k)
\leq 1 - \frac{K-1}{N}.
\]
Therefore, if
\[
\delta > 1 - \frac{K-1}{N},
\]
then the above bound is violated, and data secrecy cannot be guaranteed. In particular, there exists at least one user that can recover information beyond its requested linear combination. This completes the proof of Theorem~\ref{Theorem_necessary_condition}.

\section{Converse Proof of Theorem~\ref{sufficient_condition}}
First, we show that matrix $\mathbf{D}$ satisfies the correctness condition~\eqref{correctness_condition}, and then we prove that it is necessary to guarantee data secrecy as per~\eqref{data_secrecy_condition} and~\eqref{Privacy_con}.

\textbf{Correctness:}
Since the matrices $\mathbf{D}$ and $\mathbf{E}$ are obtained from a (possibly non-secure) multi-user distributed computing scheme, they satisfy matrix factorization $\mathbf{F} = \mathbf{D}\mathbf{E}$. Consequently, each user can correctly reconstruct its requested linear combination, and the correctness condition~\eqref{correctness_condition} of the scheme is preserved.
 
Now, we show that Theorem~\ref{sufficient_condition} provides a necessary condition to guarantee data secrecy.

\textbf{Data Secrecy: }We prove this by contradiction. Suppose that there exists a user
$k \in [K]$ such that
$$\mathrm{Rank}\!\left(\mathbf{D} \setminus \mathbf{D}(:,\mathrm{Sup}({\mathbf{d}^{\mathsf{T}}_k}))\right)
    \leq K-2.$$ Define the reduced decoding matrix for user $k$ as $$\mathbf{D}_{\mathrm{Red},k} \triangleq \mathbf{D} \setminus \mathbf{D}(:,\mathrm{Sup}(\mathbf{d}^{\mathsf{T}}_k)).$$
Then $\mathbf{D}_{\mathrm{Red},k}^{\mathsf{T}} \in \mathbb{F}^{(N-w_{\mathsf{H}}(\mathbf{d}^{\mathsf{T}}_k))\times K}$.

By the rank--nullity theorem,
$$\mathrm{dim}\!\left(\mathrm{Null}(\mathbf{D}_{\mathrm{Red},k}^{\mathsf{T}})\right)
= K - \mathrm{Rank}(\mathbf{D}_{\mathrm{Red},k}^{\mathsf{T}})
\geq K - (K-2) = 2.$$
Therefore, there exist at least two linearly independent nonzero vectors
$\mathbf{b}_k, \tilde{\mathbf{b}}_k \in \mathbb{F}^K$ such that $$\mathbf{D}_{\mathrm{Red},k}^{\mathsf{T}} \mathbf{b}_k = \mathbf{0}, \mathbf{D}_{\mathrm{Red},k}^{\mathsf{T}} \tilde{\mathbf{b}}_k = \mathbf{0},$$
or, equivalently, $$\mathbf{b}_k^{\mathsf{T}} \mathbf{D}_{\mathrm{Red},k} = \mathbf{0}^{\mathsf{T}},
    \tilde{\mathbf{b}}_k^{\mathsf{T}} \mathbf{D}_{\mathrm{Red},k} = \mathbf{0}^{\mathsf{T}}.$$

Define $$\mathbf{a}_k^{\mathsf{T}} \triangleq \mathbf{b}_k^{\mathsf{T}} \mathbf{D}, \tilde{\mathbf{a}}_k^{\mathsf{T}} \triangleq \tilde{\mathbf{b}}_k^{\mathsf{T}} \mathbf{D}.$$
Since $\mathbf{b}_k$ and $\tilde{\mathbf{b}}_k$ are linearly independent and $\mathbf{D}$ is full rank, the vectors
$\mathbf{a}_k^{\mathsf{T}}$ and $\tilde{\mathbf{a}}_k^{\mathsf{T}}$ are linearly
independent. Moreover,
$$
\mathrm{Sup}(\mathbf{a}_k^{\mathsf{T}}) ,\mathrm{Sup}(\tilde{\mathbf{a}}_k^{\mathsf{T}}) \subseteq \mathrm{Sup}(\mathbf{d}^{\mathsf{T}}_k),
$$
which implies that the user $k$ can compute linear combinations
$
\tilde{\mathbf{a}}_k^{\mathsf{T}} \mathbf{E}\mathbf{w}
$ and
$
\mathbf{a}_k^{\mathsf{T}} \mathbf{E}\mathbf{w}
$ from the same set of received responses $\mathcal{A}_k$. Since the submatrix $\mathbf{E(\mathrm{Sup(\mathbf{d}_k^{\mathsf{T}}),:})}$ is full rank, $
\tilde{\mathbf{a}}_k^{\mathsf{T}} \mathbf{E}\mathbf{w}
$ and
$
\mathbf{a}_k^{\mathsf{T}} \mathbf{E}\mathbf{w}
$
are linearly independent. So, user $k$ can learn information beyond its
request, violating the data secrecy constraints~\eqref{data_secrecy_condition} and~\eqref{Privacy_con}. Therefore, to guarantee data secrecy, it is necessary that
$$
\mathrm{Rank}\!\left(\mathbf{D} \setminus \mathbf{D}(:,\mathrm{Sup}({\mathbf{d}^{\mathsf{T}}_k}))\right)
\geq K-1.$$
\label{proof_ds_finite}

\section{Proof of Information Leakage over $\mathbb{R}$ of Theorem~\ref{sufficient_condition}}
\label{proof_ds_real}
To guarantee data secrecy in computing over the real field, we show using~\eqref{Privacy_con} that for any $\varepsilon>0$, there exists a variance $\sigma_c^2$ for the i.i.d.\ Gaussian common randomness $\mathcal{C}$, with $C_i\sim\mathcal{N}(0,\sigma_c^2)$, such that for every user $k\in[K]$,
$$I\!\left(\mathcal{W};\mathcal{A}_k,\mathbf{f}_k \mid \langle \mathbf{f}_k,\mathbf{w}\rangle\right)\leq \varepsilon .$$

From~\eqref{query_indepence}, we have
\begin{align}
    I\!\left(\mathcal{W};\mathcal{A}_k,\mathbf{f}_k \mid \langle \mathbf{f}_k,\mathbf{w}\rangle\right)
&= I\!\left(\mathcal{W};\mathcal{A}_k \mid \mathbf{f}_k,\langle \mathbf{f}_k,\mathbf{w}\rangle\right)\nonumber\\
&= h(\mathcal{A}_k \mid \mathbf{f}_k,\langle \mathbf{f}_k,\mathbf{w}\rangle)\nonumber\\
  & - h(\mathcal{A}_k \mid \mathbf{f}_k,\mathcal{W}).
  \label{proof_leakage_entrpy}
\end{align}

We now compute the two terms in~\eqref{proof_leakage_entrpy}.

\textbf{Step 1 (Computing $h(\mathcal{A}_k \mid \mathbf{f}_k,\langle \mathbf{f}_k,\mathbf{w}\rangle)$):}
Without loss of generality, assume $\mathrm{Sup}(\mathbf{d}^{\mathsf{T}}_k)=[w_{\mathsf{H}}(\mathbf{d}^{\mathsf{T}}_k)]$.
Applying the chain rule for differential entropy yields 
\begin{align}
    h(\mathcal{A}_k \mid \mathbf{f}_k,\langle \mathbf{f}_k,\mathbf{w}\rangle)
&= h\!\left(A_1 \mid \mathcal{A}_k\setminus A_1, \mathbf{f}_k,\langle \mathbf{f}_k,\mathbf{w}\rangle\right)\nonumber\\
&+h\!\left(\mathcal{A}_k\setminus A_1 \mid \mathbf{f}_k,\langle \mathbf{f}_k,\mathbf{w}\rangle\right).
\end{align}
Since the correctness condition~\eqref{correctness_condition} implies one linear relation among the answers $\mathcal{A}_k$, without loss of generality, the first response is fully determined given the other ones and the desired function, and hence $$h\!\left(A_1 \mid A_2,\ldots,A_{w_{\mathsf{H}}(\mathbf{d}^{\mathsf{T}}_k)},\mathbf{f}_k,\langle \mathbf{f}_k,\mathbf{w}\rangle\right)=0.$$
Since conditioning does not increase differential entropy, we obtain $$ h(\mathcal{A}_k \mid \mathbf{f}_k,\langle \mathbf{f}_k,\mathbf{w}\rangle)= h\!\left(\mathcal{A}_k\setminus A_1 \mid \mathbf{f}_k,\langle \mathbf{f}_k,\mathbf{w}\rangle\right)
\leq h\!\left(\mathcal{A}_k\setminus A_1\right).$$

From~\eqref{answer_func}, each server response $A_n$ is a linear combination of independent message symbols with variance $\sigma_w^2$ in $\mathbf{w}=[W_1, W_2,\ldots, W_L]^{\mathsf{T}}$
and independent Gaussian random variables with variance $\sigma_c^2$ in $\mathbf{c}=[C_1, C_2,\ldots, C_{N-K}]^{\mathsf{T}}$.
Therefore,
\begin{align}
    h(\mathcal{A}_k\setminus A_1)= h(\mathbf{X}_k\mathbf{w}+\mathbf{Y}_{k}\mathbf{c}),
\end{align}
where
\begin{align}
    \mathbf{X}_{k} = \mathbf{E}([2:w_{\mathsf{H}}(\mathbf{d}^{\mathsf{T}}_k)],:), \mathbf{Y}_k = \mathbf{C}([2:w_{\mathsf{H}}(\mathbf{d}^{\mathsf{T}}_k)],:).
\end{align}
 Since $\mathbf{w}$ and $\mathbf{c}$ are independent random vectors, we have
\begin{align}
    \Sigma_{\mathbf{X}_k\mathbf{w}+\mathbf{Y}_k\mathbf{c}}
    &=\mathrm{Cov}(\mathbf{X}_k\mathbf{w}+\mathbf{Y}_k\mathbf{c})\nonumber\\
    &= \mathrm{Cov}(\mathbf{X}_k\mathbf{w})+\mathrm{Cov}(\mathbf{Y}_k\mathbf{c})\nonumber\\
    &=\sigma^2_w\mathbf{X}_k\mathbf{X}_k^{\mathsf{T}}+ \sigma^2_c\mathbf{Y}_k\mathbf{Y}_k^{\mathsf{T}}.
\end{align}
Since $\mathbf{X}_k\mathbf{w}+\mathbf{Y}_k\mathbf{c}$ is a Gaussian random vector, we have
\begin{align}
    h(\mathcal{A}_k) &= \frac{1}{2}\log((2\pi e)^{w_{\mathsf{H}}(\mathbf{d}^{\mathsf{T}}_k)-1}|\Sigma_{\mathbf{X}_k\mathbf{w}+\mathbf{Y}_k\mathbf{c}}|)\nonumber\\
    &= \frac{1}{2}\log((2\pi e)^{w_{\mathsf{H}}(\mathbf{d}^{\mathsf{T}}_k)-1}|\sigma^2_w\mathbf{X}_k\mathbf{X}_k^{\mathsf{T}}+ \sigma^2_c\mathbf{Y}_k\mathbf{Y}_k^{\mathsf{T}}|),
\end{align}
where $|\Sigma_{\mathbf{X}_k\mathbf{w}+\mathbf{Y}_k\mathbf{c}}|$ is the determinant of matrix $\Sigma_{\mathbf{X}_k\mathbf{w}+\mathbf{Y}_k\mathbf{c}}.$
\textbf{Step 2 (Computing $h(\mathcal{A}_k \mid \mathbf{f}_k,\mathcal{W})$):} 
To proceed, we first establish a lower bound on the rank of the visible common randomness for each user.

\begin{Lemma}
\label{Lemma_rank_c}
In the proposed general scheme in Section~\ref{Achievable_algorithm}, for each user $k \in [K]$, we have
\begin{align}
\mathrm{Rank}\!\big(\mathbf{C}(\mathrm{Sup}(\mathbf{d}_k^{\mathsf{T}}),:)\big)
\geq w_{\mathsf{H}}(\mathbf{d}_k^{\mathsf{T}})-1.
\end{align}
\end{Lemma}

\begin{IEEEproof}
We prove by contradiction. Suppose that
\[
\mathrm{Rank}\!\big(\mathbf{C}(\mathrm{Sup}(\mathbf{d}_k^{\mathsf{T}}),:)\big)
< w_{\mathsf{H}}(\mathbf{d}_k^{\mathsf{T}})-1.
\]
Since 
\[
\mathbf{C}(\mathrm{Sup}(\mathbf{d}_k^{\mathsf{T}}), :) \in 
\mathbb{F}^{w_{\mathsf{H}}(\mathbf{d}_k^{\mathsf{T}})\times (N-K)},
\]
there exist two linearly independent vectors 
$\mathbf{a}_k^{\mathsf{T}}$ and $\tilde{\mathbf{a}}_k^{\mathsf{T}}$, both supported on $\mathrm{Sup}(\mathbf{d}_k^{\mathsf{T}})$, such that
\[
\mathbf{a}_k^{\mathsf{T}}\mathbf{C}(\mathrm{Sup}(\mathbf{d}_k^{\mathsf{T}}),:)
=
\tilde{\mathbf{a}}_k^{\mathsf{T}}\mathbf{C}(\mathrm{Sup}(\mathbf{d}_k^{\mathsf{T}}),:)
= \mathbf{0}^{\mathsf{T}}.
\]

Extend these vectors to $\mathbb{F}^N$ by zero-padding outside $\mathrm{Sup}(\mathbf{d}_k^{\mathsf{T}})$, yielding $\mathbf{b}_k^{\mathsf{T}}$ and $\tilde{\mathbf{b}}_k^{\mathsf{T}}$. Then
\[
\mathbf{b}_k^{\mathsf{T}}\mathbf{C} = \tilde{\mathbf{b}}_k^{\mathsf{T}}\mathbf{C} = \mathbf{0}^{\mathsf{T}},
\]
which implies that $\mathbf{b}_k, \tilde{\mathbf{b}}_k \in \mathrm{Null}(\mathbf{C}^{\mathsf{T}})$.

From the construction of the scheme (Step~3), the columns of $\mathbf{C}$ span $\mathrm{Null}(\mathbf{D})$, hence
\[
\mathrm{Null}(\mathbf{C}^{\mathsf{T}}) = \mathrm{Span}(\mathbf{D}).
\]

Therefore, $\mathbf{b}_k$ and $\tilde{\mathbf{b}}_k$ lie in the row space of $\mathbf{D}$. Since they are linearly independent and supported only on $\mathrm{Sup}(\mathbf{d}_k^{\mathsf{T}})$, this implies that
\begin{align}
\mathrm{Rank}\!\left(
\mathbf{D} \setminus \mathbf{D}(:,\mathrm{Sup}(\mathbf{d}_k^{\mathsf{T}}))
\right)
< K - 1,
\end{align}
which contradicts the condition imposed in Step~1 of the general scheme. 

Therefore,
\[
\mathrm{Rank}\!\big(\mathbf{C}(\mathrm{Sup}(\mathbf{d}_k^{\mathsf{T}}),:)\big)
\geq w_{\mathsf{H}}(\mathbf{d}_k^{\mathsf{T}})-1,
\]
which completes the proof.
\end{IEEEproof}

We now compute $h(\mathcal{A}_k \mid \mathbf{f}_k,\mathcal{W})$. Conditioned on $\mathcal{W}$, the responses depend only on the common randomness, hence
$$h(\mathcal{A}_k \mid \mathbf{f}_k,\mathcal{W})
= h\!\left(\mathbf{C}([w_{\mathsf{H}}(\mathbf{d}^{\mathsf{T}}_k)],:)\mathbf{c}\right).$$
By Lemma~\ref{Lemma_rank_c},
$\mathrm{Rank}(\mathbf{C}(\mathrm{Sup}(\mathbf{d}^{\mathsf{T}}_k),:))\geq w_{\mathsf{H}}(\mathbf{d}^{\mathsf{T}}_k)-1$,
which implies
\begin{align}
h\!\left(\mathbf{C}([w_{\mathsf{H}}(\mathbf{d}^{\mathsf{T}}_k)],:)\mathbf{c}\right)&\geq h\!\left(\mathbf{C}([2:w_{\mathsf{H}}(\mathbf{d}^{\mathsf{T}}_k)],:)\mathbf{c}\right)\nonumber\\
& = h\!\left(\mathbf{Y}_k\mathbf{c}\right).
\end{align}

Since $\mathbf{Y}_k \mathbf{c}$ is a Gaussian vector, we obtain
\begin{align}
h(\mathcal{A}_k \mid \mathbf{f}_k,\mathcal{W})
\geq \frac{1}{2}\log\!\Big((2\pi e)^{w_{\mathsf{H}}(\mathbf{d}_k^{\mathsf{T}})-1}
\big|\sigma_c^2 \mathbf{Y}_k \mathbf{Y}_k^{\mathsf{T}}\big|\Big).
\end{align}

\textbf{Step 3 (Information leakage bound):}
Combining Steps~1 and~2 yields
\begin{align}
    &I\!\left(\mathcal{W};
\mathcal{A}_k,\mathbf{f}_k \mid \langle \mathbf{f}_k,\mathbf{w}\rangle\right)\nonumber\\
&\leq \frac{1}{2}\log((2\pi e)^{w_{\mathsf{H}}(\mathbf{d}^{\mathsf{T}}_k)-1}|\sigma^2_w\mathbf{X}_k\mathbf{X}_k^{\mathsf{T}}+ \sigma^2_c\mathbf{Y}_k\mathbf{Y}_k^{\mathsf{T}}|)\nonumber\\ &- \frac{1}{2}\log((2\pi e)^{w_{\mathsf{H}}(\mathbf{d}^{\mathsf{T}}_k)-1}|\sigma^2_c\mathbf{Y}_k\mathbf{Y}_k^{\mathsf{T}}|)\nonumber\\ &= \frac{1}{2}\log(|\frac{\sigma^2_w\mathbf{X}_k\mathbf{X}_k^{\mathsf{T}}+ \sigma^2_c\mathbf{Y}_k\mathbf{Y}_k^{\mathsf{T}}}{\sigma^2_c\mathbf{Y}_k\mathbf{Y}_k^{\mathsf{T}}}|)
\end{align}
Note that $\mathbf{Y}_k = \mathbf{C}([2:w_{\mathsf{H}}(\mathbf{d}^{\mathsf{T}}_k)],:)$ is full rank, so $\mathbf{Y}_k\mathbf{Y}_k^{\mathsf{T}}$ is invertible, thus
\begin{align}
&I\!\left(\mathcal{W};
\mathcal{A}_k,\mathbf{f}_k \mid \langle \mathbf{f}_k,\mathbf{w}\rangle\right)\nonumber\\
&\leq \frac{1}{2}\log(|\mathbf{I}_{{w_{\mathsf{H}}(\mathbf{d}^{\mathsf{T}}_k)-1}}+\frac{\sigma_w^2}{\sigma_c^2}\mathbf{X}_k\mathbf{X}_k^{\mathsf{T}}(\mathbf{Y}_k\mathbf{Y}_k^{\mathsf{T}})^{-1}|).
\end{align}
The determinant $|\mathbf{I}_{{w_{\mathsf{H}}(\mathbf{d}^{\mathsf{T}}_k)-1}}+\frac{\sigma_w^2}{\sigma_c^2}\mathbf{X}_k\mathbf{X}_k^{\mathsf{T}}(\mathbf{Y}_k\mathbf{Y}_k^{\mathsf{T}})^{-1}|$ is equal to $\prod\limits_{i=1}^{w_{\mathsf{H}}(\mathbf{d}^{\mathsf{T}}_k)-1}(1+\lambda_i(\frac{\sigma_w^2}{\sigma_c^2}\mathbf{X}_k\mathbf{X}_k^{\mathsf{T}}(\mathbf{Y}_k\mathbf{Y}_k^{\mathsf{T}})^{-1})),$ where for each $i \in [1:{{w_{\mathsf{H}}(\mathbf{d}^{\mathsf{T}}_k)}}-1]$, $\lambda_i(\frac{\sigma_w^2}{\sigma_c^2}\mathbf{X}_k\mathbf{X}_k^{\mathsf{T}}(\mathbf{Y}_k\mathbf{Y}_k^{\mathsf{T}})^{-1})$ is the eigenvalue of the matrix $\frac{\sigma_w^2}{\sigma_c^2}\mathbf{X}_k\mathbf{X}_k^{\mathsf{T}}(\mathbf{Y}_k\mathbf{Y}_k^{\mathsf{T}})^{-1}.$ Thus
\begin{align}
    &I\!\left(\mathcal{W};
\mathcal{A}_k,\mathbf{f}_k \mid \langle \mathbf{f}_k,\mathbf{w}\rangle\right)\nonumber\\
&\leq \frac{1}{2}\sum\limits_{i=1}^{{w_{\mathsf{H}}(\mathbf{d}^{\mathsf{T}}_k)-1}}\log(1+\lambda_i(\frac{\sigma_w^2}{\sigma_c^2}\mathbf{X}_k\mathbf{X}_k^{\mathsf{T}}(\mathbf{Y}_k\mathbf{Y}_k^{\mathsf{T}})^{-1}))\nonumber\\
&\leq \frac{{w_{\mathsf{H}}(\mathbf{d}^{\mathsf{T}}_k)-1}}{2}\log(1+\frac{\sigma_w^2}{\sigma_c^2}\lambda_{\mathrm{max}}(\mathbf{X}_k\mathbf{X}_k^{\mathsf{T}}(\mathbf{Y}_k\mathbf{Y}_k^{\mathsf{T}})^{-1})),
\end{align}
where $\lambda_{\mathrm{max}}$ is the maximum eigenvalue. Since $\mathbf{X}_k\mathbf{X}_k^{\mathsf{T}}$ and $\mathbf{Y}_k\mathbf{Y}_k^{\mathsf{T}}$ are positive semi-definite, we have
\begin{align}
    &I\!\left(\mathcal{W};
\mathcal{A}_k,\mathbf{f}_k \mid \langle \mathbf{f}_k,\mathbf{w}\rangle\right)\nonumber\\
&\leq \frac{{w_{\mathsf{H}}(\mathbf{d}^{\mathsf{T}}_k)-1}}{2}\log(1+\frac{\sigma_w^2}{\sigma_c^2}\lambda_{\mathrm{max}}(\mathbf{X}_k\mathbf{X}_k^{\mathsf{T}}(\mathbf{Y}_k\mathbf{Y}_k^{\mathsf{T}})^{-1}))\nonumber\\
&\leq \frac{{w_{\mathsf{H}}(\mathbf{d}^{\mathsf{T}}_k)-1}}{2}\log(1+\frac{\sigma_w^2}{\sigma_c^2}\lambda_{\mathrm{max}}(\mathbf{X}_k\mathbf{X}_k^{\mathsf{T}})\lambda_{\mathrm{max}}((\mathbf{Y}_k\mathbf{Y}_k^{\mathsf{T}})^{-1})))\nonumber\\
&= \frac{{w_{\mathsf{H}}(\mathbf{d}^{\mathsf{T}}_k)-1}}{2}\log(1+\frac{\sigma_w^2}{\sigma_c^2}\frac{ \lambda_{\mathrm{max}}(\mathbf{X}_k\mathbf{X}_k^{\mathsf{T}})}{\lambda_{\mathrm{min}}(\mathbf{Y}_k\mathbf{Y}_k^{\mathsf{T}})}),
\end{align}

Since $\mathbf{Y}_k\mathbf{Y}_k^{\mathsf{T}}$ is invertible, we have 
$\lambda_{\mathrm{min}}(\mathbf{Y}_k\mathbf{Y}_k^{\mathsf{T}}) \neq 0$. Therefore, for any $\varepsilon > 0$, one can choose $\sigma_c^2$ to be sufficiently large such that the information leakage to the user $k$ is bounded by $\varepsilon$. 

Moreover, the same argument applies to any index set $\mathcal{S}_k$ corresponding to a maximal linearly independent subset of rows of $\mathbf{C}(\mathrm{Sup}(\mathbf{d}_k^{\mathsf{T}}),:)$. 

This establishes the data secrecy condition~\eqref{Privacy_con} and completes the proof of the information leakage bound~\eqref{real_data_leakage} over $\mathbb{R}$ in Theorem~\ref{sufficient_condition}.

\end{appendices}
\bibliographystyle{IEEEtran}
\bibliography{references}
\end{document}